\documentclass[a4paper]{article}
\usepackage{amsmath}
\usepackage{amssymb} 
\usepackage{amsthm} 
\usepackage{graphicx}
\usepackage[font=small, format=hang, labelfont=bf]{caption}
\usepackage[subrefformat=parens, labelfont=default]{subfig}
\usepackage{a4wide}
\usepackage[pdfpagelabels, pagebackref, naturalnames]{hyperref}
\usepackage{path} 
\usepackage{wrapfig} 
\usepackage{enumerate} 
\usepackage{color}

\title{Untangling a Planar Graph\thanks{This paper is based on two
    preliminary versions; (a)~X.~Goaoc, J.~Kratochv\'{\i}l,
    Y.~Okamoto, Ch.-S.~Shin, and A.~Wolff: ``Moving Vertices to Make
    Drawings Plane'', Proc.\ 15th~International Symposium on Graph
    Drawing (GD'07), volume 4875 of \emph{Lecture Notes in Computer
      Science}, pages 101--112, Springer-Verlag 2008, and
    (b)~A.~Spillner and A.~Wolff: ``Untangling a Planar Graph'',
    Proc.\ 34th~International Conference on Current Trends in Theory
    and Practice of Computer Science (SOFSEM'08), volume 4910 of
    \emph{Lecture Notes in Computer Science}, pages 473--484,
    Springer-Verlag 2008.
    Our work was started at the 9th ``Korean'' Workshop on Computational
    Geometry and Geometric Networks organized by A.~Wolff and X.~Goaoc,
    July~30--August~4, 2006 in Schloss Dagstuhl, Germany.  Further
    contributions were made at the 2nd Workshop on Graph Drawing and
    Computational Geometry organized by W.~Didimo and G.~Liotta, March
    11--16, 2007 in Bertinoro, Italy.}}

\author{%
  Xavier~Goaoc\thanks{LORIA~-- INRIA Grand Est, Nancy, France. Email:
    \email{goaoc@loria.fr}}
  \and
  Jan~Kratochv{\'i}l\thanks{Department of Applied Mathematics and
    Institute of Theoretical Computer Science, Charles University,
    Czech~Republic. Email: \email{honza@kam.mff.cuni.cz}.  Supported
    by research project 1M0545 of the Czech Ministery of Education.}
  \and
  Yoshio~Okamoto\thanks{Graduate School of Information Science and
    Engineering, Tokyo Institute of Technology, Japan.
    Email: \email{okamoto@is.titech.ac.jp}.  
    Supported by
    GCOE Program ``Computationism as a Foundation for the Sciences''
    and Grant-in-Aid for Scientific Research from Ministry of Education,
    Science and Culture, Japan, and Japan Society for the Promotion of
    Science.}
  \and
  Chan-Su~Shin\thanks{School of Electronics and Information Engineering,
    Hankuk University of Foreign Studies, Yongin, Korea. Email:
    \email{cssin@hufs.ac.kr}.  Supported by Research Grant 2008 of the
    Hankuk University of Foreign Studies.}
  \and
  Andreas Spillner\thanks{School of Computing Sciences,
  University of East Anglia, Norwich, United Kingdom.  Email:
  \email{aspillner@cmp.uea.ac.uk}}
  \and
  Alexander~Wolff\thanks{Faculteit Wiskunde en Informatica, Technische
    Universiteit Eindhoven, The~Netherlands. WWW:
    \email{http://www.win.tue.nl/\~{}awolff}} }

\date{}

\graphicspath{{pics/}}

\renewcommand*{\backref}[1]{}
\renewcommand*{\backrefalt}[4]{%
  \small
  \ifcase #1 %
  [Not cited.]%
  \or
  [Cited on page~#2.]%
  \else
  [Cited on pages~#2.]
  \fi
}

\hypersetup{%
  pdftitle={Untangling a Planar Graph}, 
  pdfauthor={}, 
  pdfborder={0 0 1},
  pdfcreator={}, pdfproducer={},
  citebordercolor={0 .667 0},          
  linkbordercolor={.5812 .0665 .0659}, 
  urlbordercolor={0 0 .667}            
}

\theoremstyle{plain}
\newtheorem{definition}{Definition}[section]
\newtheorem{theorem}[definition]{Theorem}
\newtheorem{lemma}[definition]{Lemma}
\newtheorem{observation}[definition]{Observation}
\newtheorem{proposition}[definition]{Proposition}
\newtheorem{corollary}[definition]{Corollary}

\theoremstyle{definition}
\newtheorem{remark}[definition]{Remark}

\newenvironment{keywords}{\vspace*{4ex}\centering\textbf{Keywords:}\hspace{1ex}}{\vspace*{3ex}\normalsize}

\renewcommand{\figurename}{Fig.}
\newcommand{\email}[1]{\texttt{#1}}

\newcommand{\pspace}{\ensuremath{\cal PSP\!A\,CE}}
\newcommand{\Vtop}{\ensuremath{V_\mathrm{top}}}
\newcommand{\Vbot}{\ensuremath{V_\mathrm{bot}}}
\newcommand{\Htop}{\ensuremath{\mathcal{H}_\mathrm{top}}}
\newcommand{\Gtop}{\ensuremath{G^+_\mathrm{top}}}
\newcommand{\Gbot}{\ensuremath{G^+_\mathrm{bot}}}

\newcommand{\R}{\ensuremath{\mathbb{R}}}
\newcommand{\eps}{\ensuremath{\varepsilon}}

\newcommand{\shift}{\ensuremath{\mathrm{shift}}}
\newcommand{\fix}{\ensuremath{\mathrm{fix}}}

\newcommand{\clau}{\ensuremath{C}}
\newcommand{\G}{\ensuremath{G}}
\newcommand{\HH}{\ensuremath{H}}
\newcommand{\dG}{\ensuremath{\delta_G}}
\newcommand{\dH}{\ensuremath{\delta_H}}

\renewcommand{\bottomfraction}{.99}

\begin{document}

\maketitle

\begin{abstract}
  A straight-line drawing~$\delta$ of a planar graph~$G$ need not be
  plane, but can be made so by \emph{untangling} it, that is, by
  moving some of the vertices of~$G$.  Let $\shift(G,\delta)$ denote the
  minimum number of vertices that need to be moved to
  untangle~$\delta$.  We show that $\shift(G,\delta)$ is NP-hard to
  compute and to approximate.  Our hardness results extend to a
  version of \textsc{1BendPointSetEmbeddability}, a well-known
  graph-drawing problem.

  Further we define $\fix(G,\delta)=n-\shift(G,\delta)$ to be the
  maximum number of vertices of a planar $n$-vertex graph~$G$ that can
  be fixed when untangling~$\delta$.  We give an algorithm that fixes
  at least $\sqrt{((\log n)-1)/\log \log n}$ vertices when untangling
  a drawing of an $n$-vertex graph~$G$.  If~$G$ is outerplanar, the
  same algorithm fixes at least $\sqrt{n/2}$ vertices.  On the other
  hand we construct, for arbitrarily large~$n$, an $n$-vertex planar
  graph~$G$ and a drawing~$\delta_G$ of~$G$ with $\fix(G,\delta_G) \le
  \sqrt{n-2}+1$ and an $n$-vertex outerplanar graph~$H$ and a
  drawing~$\delta_H$ of~$H$ with $\fix(H,\delta_H) \le 2
  \sqrt{n-1}+1$.  Thus our algorithm is asymptotically worst-case
  optimal for outerplanar graphs.
\end{abstract}

\begin{keywords}
   Graph drawing, straight-line drawing, planarity, NP-hardness,
   hardness of approximation, moving vertices, untangling,
   point-set embeddability.
\end{keywords}

\section{Introduction}
\label{sec:introduction}

A \emph{drawing} of a graph~$G$ maps each vertex of~$G$ to a distinct
point of the plane and each edge~$uv$ to an open Jordan curve
connecting the images of~$u$ and~$v$.  A drawing of~$G$ is
\emph{plane} if no two distinct edges \emph{cross}, that is, intersect.
By the famous theorem of
Wagner~\cite{w-b4-36}, F{\'a}ry~\cite{f-slrpg-48}, and
Stein~\cite{s-cm-51}, any planar graph admits a plane \emph{straight-line}
drawing, that is, a drawing that maps edges to straight-line segments.  
Obviously not every straight-line drawing of a planar graph
is plane.  In this paper we are exclusively interested in such 
straight-line drawings.
Thus by a drawing we will always mean a straight-line drawing.
Since a (straight-line) drawing is completely defined by the position
of the vertices, moving a vertex is a natural operation to modify such
a drawing.  If a drawing is to be made plane---or
\emph{untangled}---by successively moving vertices, it is desirable to
move as few vertices as possible.  The smaller the number of moves, the
less likely it is that an observer gets confused, that is, the more
likely the observer's \emph{mental map}~\cite{plms-lamm-95} is
preserved during a sequence of changes.  A recreational version of the
problem of minimizing the number of moves is given by Tantalo's
popular on-line game \emph{Planarity}~\cite{t-p-07}, where the aim is
to untangle a straight-line drawing as quickly as possible, again by
vertex moves.  Actually, in Tantalo's game an additional difficulty
for the player is the fixed size of the screen; Liske's \cite{l-p-05}
version of the game allows rescaling and hence is fully equivalent to
untangling.

We define the \emph{vertex-shifting distance}~$d$ between two
drawings~$\delta$ and~$\delta'$ of a graph~$G=(V,E)$ to be the number of
vertices of~$G$ whose images under~$\delta$ and~$\delta'$ differ:
\[ d(\delta, \delta') = \big|\{v \in V \mid \delta(v) \neq
\delta'(v)\}\big|.\]
Given our edit operation, $d$ represents the edit distance for
straight-line drawings of graphs (see \figurename~\ref{fig:first-example}
for an example). For a drawing~$\delta$ of a planar graph~$G$, we denote
by $\shift(G,\delta)$ the minimum number of vertices that need to be
moved in order to untangle~$\delta$.  In some sense $\shift(G,\delta)$
measures the distance of~$\delta$ from planarity.  This suggests the
following computational problem.
\begin{quote}
  \textsc{MinShiftedVertices}$(G,\delta)$: given a drawing~$\delta$ of a
  planar graph~$G$, find a plane drawing~$\delta'$ of~$G$ with
  $d(\delta,\delta') = \shift(G,\delta)$.
\end{quote}
The symmetric point of view is often helpful.  Therefore we denote by
$\fix(G,\delta)$ the maximum number of vertices that can be fixed when
untangling~$\delta$; we refer to such vertices as \emph{fixed}
vertices.  Clearly, $\fix(G,\delta) = n -\shift(G,\delta)$, where~$n$ is
the number of vertices of~$G$.  We call the corresponding problem,
that is, finding a plane drawing of a given planar graph~$G$ that
maximizes the number of fixed vertices with a given drawing~$\delta$,
\textsc{MaxFixedVertices}.  We denote by $\fix(G)$ the minimum of
$\fix(G, \delta)$ over all drawings~$\delta$ of~$G$.  Analogously, we
denote by~$\shift(G)$ the maximum of $\shift(G,
\delta)$ over all drawings~$\delta$ of~$G$.

\begin{wrapfigure}[11]{r}{4.2cm}
  \centering %
  \includegraphics{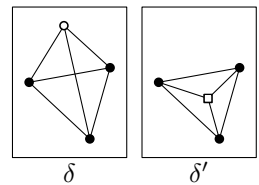}
  \caption{Two drawings of~$K_4$: $\delta$~is not plane, $\delta'$~is
    plane; $d(\delta,\delta')=1$.}
  \label{fig:first-example}
\end{wrapfigure}

Kaufmann and Wiese~\cite{kw-evpfb-02} considered the graph-drawing
problem \textsc{1BendPointSetEmbeddability} that will turn out to be
related to \textsc{MinShiftedVertices}.  They defined a planar graph
$G=(V,E)$ to be \emph{$k$-bend embeddable} if, for any set~$S$
of~$|V|$ points in the plane, there is a one-to-one correspondence 
between~$V$ and~$S$ that can be extended to a plane drawing of~$G$
with at most~$k$ bends per edge.  Kaufmann and Wiese
showed that (a)~every four-connected planar graph is 1-bend embeddable,
(b)~every planar graph is 2-bend embeddable, and (c)~given a planar
graph $G=(V,E)$ and set~$S$ of~$|V|$ points on a line, it is NP-complete
to decide whether there is a correspondence between~$V$ and~$S$ that
makes it possible to 1-bend embed~$G$ on~$S$.

The contributions we present in this paper are three-fold:

\begin{itemize}
\item We prove that the decision versions of \textsc{MaxFixedVertices}
  and \textsc{MinShiftedVertices} are NP-hard
  (Theorem~\ref{thm:hardness-mmv}) and lie in \pspace\
  (Proposition~\ref{pro:vertexmove-pspace}).  We further prove that 
  \textsc{MinShiftedVertices} is hard to approximate in the following
  sense: if there is a real $\eps \in (0,1]$ and a
  polynomial-time algorithm that guarantees to untangle any
  drawing~$\delta$ of any $n$-vertex planar graph~$G$ with at most 
  $(n^{1-\eps}) \cdot (\shift(G,\delta)+1)$ moves, then ${\cal P}
  = \cal NP$ (Theorem~\ref{thm:non-approx-mmv}).  

\item We complement the complexity result of Kaufmann and
  Wiese~\cite{kw-evpfb-02} on \textsc{1BendPointSetEmbeddability} by
  showing that it is NP-hard to decide whether a \emph{given} one-to-one
  correspondence between the vertices of a planar graph~$G$ and a
  planar point set~$S$ extends into a plane drawing of~$G$ with at
  most one bend per edge (Theorem~\ref{thm:hardness-one-bend}). We
  also show that the problem lies in \pspace\
  (Theorem~\ref{thm:1bend-pspace}) and that an optimization version of
  the problem is hard to approximate
  (Corollary~\ref{cor:non-approx-one-bend}).

\item We show that 
  $\fix(H) \geq \sqrt{n/2}$ for any $n$-vertex outerplanar graph~$H$
  (Corollary~\ref{cor:outerplanar}) and that $\fix(G) \ge
  \sqrt{\frac{(\log n) -1}{\log \log n}}$ for any general planar
  graph~$G$ with $n \ge 4$ vertices
  (Theorem~\ref{theorem:final:lower:bound}), where the base of
  logarithms is~$2$.  We also give, for arbitrarily large~$n$,
  examples of an $n$-vertex outerplanar graph~$H$ with $\fix(H) \le 2
  \sqrt{n-1} + 1$ (Theorem~\ref{thm:outerplanar}) and of an $n$-vertex
  planar graph~$G$ with $\fix(G) \le \sqrt{n-2}+1$
  (Theorem~\ref{thm:planar-upper}).  We prove the two bounds by
  using drawings where all vertices lie on a line.
\end{itemize}

\section{Previous and Related Work}

Arguably, one of the earliest results on untangling, for the $n$-path
in the real line, is the Erd\H{o}s--Szekeres theorem, which we state
here for further reference.

\begin{theorem}[Erd\H{o}s and Szekeres~\cite{es-cpg-35}]
  \label{erdos-szekeres} %
  Any sequence of~$n \ge sr+1$ different real numbers has an
  increasing subsequence of~$s+1$ terms or a decreasing subsequence
  of~$r+1$ terms. 
\end{theorem}

The current best bounds on $\fix(G)$, where~$G$ is restricted to
certain classes of planar graphs, are summarized in
\tablename~\ref{table:summary}.  Recall 
that a lower bound of~$f(n)$ means that we can untangle \emph{any}
drawing of \emph{any} $n$-vertex graph~$G$ in the given graph class
while fixing at least~$f(n)$ vertices, whereas an upper bound
of~$g(n)$ means that for arbitrarily large~$n$ \emph{there exists} a
drawing~$\delta$ of an $n$-vertex graph~$G$ in the given graph class
such that at most~$g(n)$ vertices can stay fixed when
untangling~$\delta$.

\begin{table}[ht]
  \centering
  {\def\arraystretch{1.1}
  \def\tabcolsep{2ex}
  \begin{tabular}{@{}l|c|c|c|c@{}}
    Graph class     & \multicolumn{2}{c|}{Lower bound} &
    \multicolumn{2}{c}{Upper bound} \\\hline&&&&\\[-2.5ex]

    Cycles 		& $\Omega(n^{2/3})$ & \cite{c-upg-08} &
    $O((n \log n)^{2/3})$ & \cite{pt-up-02}\\
    Trees  		& $\sqrt{n/2}$ &
    \cite{bdhlmw-pbugp-08}
    & $3\sqrt{n}-3$ & \cite{bdhlmw-pbugp-08}\\ 
    Outerplanar graphs	& $\sqrt{n/2}$ & \cite{rv-cssld-08} \&
    Corollary~\ref{cor:outerplanar} 
                        & $2\sqrt{n-1}+1$ & Theorem~\ref{thm:outerplanar}\\
    Planar graphs & $\sqrt[4]{(n+1)/2}$ & \cite{bdhlmw-pbugp-08} &
    $\sqrt{n-2}+1$ & 
    Theorem~\ref{thm:planar-upper} \\
  \end{tabular}}

  \vspace{1.5ex}
  \caption{Best known bounds for $\fix(G)$, where~$G$ is a graph of
    the given graph class with~$n$ vertices.}   
  \label{table:summary}
\end{table}

Untangling was first investigated for the \emph{$n$-cycle}~$C_n$,
following the question by Watanabe~\cite{w-op-98} of whether 
$\fix(C_n) \in \Omega(n)$. The answer turned out to be negative: Pach
and Tardos~\cite{pt-up-02} showed, by a probabilistic argument, that
$\fix(C_n) \in O((n \log n)^{2/3})$. They also showed that $\fix(C_n)
\ge\lfloor \sqrt{n}+1 \rfloor$ by applying the Erd\H{o}s--Szekeres
theorem to the sequence of the indices of the vertices of~$\delta$ in
clockwise order around some specific point. Cibulka \cite{c-upg-08}
recently improved that lower bound to $\Omega(n^{2/3})$ by applying
the Erd\H{o}s--Szekeres theorem not once but $\Theta(n^{1/3})$ times.

Pach and Tardos~\cite{pt-up-02} extended the question to planar
graphs and asked whether there is a constant $\gamma > 0$ such that
$\fix(G) \in \Omega(n^\gamma)$ for any planar $n$-vertex
graph~$G$. This question was recently answered in the affirmative by
Bose et al.~\cite{bdhlmw-pbugp-08}, who showed that $\fix(G) \ge
\sqrt[4]{n/3}$.  While their bound improves on
our Theorem~\ref{theorem:final:lower:bound}, their algorithm uses our
algorithm as a subroutine (specifically the result in
Corollary~\ref{cor:outerplanar}).  A recent improvement in our
analysis also improves their bound, yielding $\fix(G) \ge
\sqrt[4]{(n+1)/2}$.  Kang et al.~\cite{kprsc-odpg-08} showed that for
arbitrarily large~$n$ there is a 
planar graphs~$G$ with~$n$ vertices and fix$(G) \le 2\sqrt{n} + 1$.  
For our upper bound of $\sqrt{n-2}+1$, see Theorem~\ref{thm:planar-upper}.
Kang et al.~\cite{kprsc-odpg-08} also shed some light on how
upper bounds on $\fix(G)$ are affected by restricting
the possible locations of vertices in the drawings of~$G$.
In particular, they showed that initial drawings
with all vertices on a line, such as our examples in
Theorems~\ref{thm:outerplanar} and~\ref{thm:planar-upper}, are the
worst case in the sense that 
any planar graph~$G$ has such a drawing~$\delta$ with $\fix(G) =
\fix(G,\delta)$, and that their upper bound holds even in the case
where initial drawings are restricted to drawings where vertices
correspond to a set of points on the boundary of a convex set.  (Note
that this generalizes both the vertices-on-a-line case and the
vertices-in-convex-position case.) 

Verbitsky~\cite{v-ocpg-08} investigated planar graphs of higher connectivity.
He proved linear upper bounds on $\fix(G)$ for three- and four-connected
planar graphs.  Cibulka~\cite{c-upg-08} gave, for any planar
graph~$G$, an upper bound on $\fix(G)$ that is a function of the
number of vertices, the maximum degree, and the diameter of~$G$. This
latter bound implies, in particular, that $\fix(G) \in O((n \log
n)^{2/3})$ for any three-connected planar graph~$G$ and that any graph~$H$ such
that $\fix(H) \ge cn$ for some $c>0$ must have a vertex of
degree $\Omega(nc^2/\log^2 n)$.

For the class of trees, Bose et al.~\cite{bdhlmw-pbugp-08}
showed
that $\fix(T) \ge \sqrt{n/2}$ for any tree~$T$ with~$n$
vertices. They further showed that $\fix(T) \le 3\sqrt{n}-3$ for a
collection of stars with~$n$ vertices in total, which, up to adding
one vertex to turn these stars into a single tree, implies that the
previous bound is asymptotically tight.  We have obtained the same
lower bound of $\sqrt{n/2}$ for the larger class of outerplanar graphs
(Corollary~\ref{cor:outerplanar}).  This bound was obtained
independently by Ravsky and Verbitsky~\cite{rv-cssld-08} via a finer
analysis of sets of collinear vertices in plane drawings.

The hardness of computing $\fix(G,\delta)$ given~$G$ and~$\delta$ was
obtained independently by Verbitsky~\cite{v-ocpg-08} by a 
reduction from independent set in line-segment intersection
graphs. While our proof is more complicated than his, it is stronger
as it also yields hardness of approximation and extends
to the problem \textsc{1BendPointSetEmbeddability} with given vertex--point
correspondence. 

Finally, a somewhat related problem is that of \emph{morphing}, or
isotopy, between two plane drawings~$\delta_1$ and~$\delta_2$ of the
same graph~$G$, that is, to define for each vertex~$v$ of~$G$ a
movement from~$\delta_1(v)$ to~$\delta_2(v)$ such that at any time
during the move the drawing defined by the current vertex positions is
plane.  We refer the interested reader to the survey by Lubiw et
al.~\cite{lps-mopgd-06}.

\section{Complexity}
\label{sec:complexity}

In this section, we investigate the complexity of
\textsc{MinShiftedVertices} and of \textsc{1BendPointSetEmbeddability}
with given vertex--point correspondence.

\begin{theorem} 
  \label{thm:hardness-mmv}
  Given a planar graph~$G$, a drawing~$\delta$ of~$G$, and an integer
  $K>0$, it is NP-hard to decide whether $\shift(G,\delta) \le K$.
\end{theorem}

\begin{proof}
  Our proof is by reduction from \textsc{Planar3SAT}, which is 
  NP-hard \cite{l-pftu-82}.  An instance of \textsc{Planar3SAT}
  is a 3-SAT formula~$\varphi$ whose variable-clause graph is planar.
  Note that this graph can be laid out (in polynomial time) such that
  variables correspond to rectangles centered on the $x$-axis and
  clauses correspond to non-crossing three-legged ``combs'' completely
  above or completely below the $x$-axis \cite{kr-pcr-92}, see
  \figurename~\ref{fig:planar}.  We refer to this layout of the
  variable-clause graph as~$\lambda_\varphi$.  We now construct a
  graph~$G_\varphi$ with a straight-line drawing~$\delta_\varphi$ such that 
  the following holds: $\delta_\varphi$ can be untangled by moving at 
  most~$K$ vertices if and only if~$\varphi$ is satisfiable.  We
  fix~$K$ later.

  \begin{figure}[ht]
    \centering
    \includegraphics{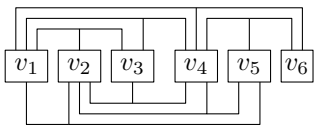}
    \caption{Embedding of a planar 3-SAT formula.}
    \label{fig:planar}
  \end{figure}

  Our graph~$G_\varphi$ consists of two types of substructures (or
  \emph{gadgets}), modeling the variables and clauses of~$\varphi$.
  The overall layout of~$G_\varphi$ follows~$\lambda_\varphi$ (see
  \figurename~\ref{fig:planar}): the variable gadgets are drawn in the
  same order along the $x$-axis as the variable nodes in~$\lambda_\varphi$,
  and the clause gadgets form non-crossing three-legged combs that lie
  on the same side of the $x$-axis as the corresponding clause nodes
  in~$\lambda_\varphi$.

  In our gadgets, see Figs.~\ref{fig:variable} and~\ref{fig:clause},
  there are two types of vertices and edges; those that \emph{may}
  move and those that are \emph{meant} not to move.  We refer to the
  two types as \emph{mobile} and \emph{immobile}.  Each mobile vertex
  (but no immobile vertex) is incident to two edges that cross two
  other edges.  The drawing~$\delta_\varphi$ that we specify in the
  following has~$2K$ crossings; if~$\varphi$ has a satisfying truth
  assignment, $\delta_\varphi$ can be untangled by moving~$K$ mobile
  vertices. Otherwise, at least one immobile vertex must move, and thus
  in total at least $K+1$ vertices need to move.
  In the figures, immobile vertices are marked by black disks, mobile vertices
  by circles, and their predestined positions
  by little squares.  
  Mobile edges---edges incident to a mobile vertex---are drawn as thick 
  solid gray line segments, and their predestined positions 
  as gray line segments that are dashed, dotted, or
  dashed-dotted (and thus not solid).  
  Immobile edges are drawn as solid black line segments.

  Now consider the gadget for some variable~$x$ in~$\varphi$, see the
  shaded area in \figurename~\ref{fig:variable}.  
  \begin{figure}
    \centering %
    \includegraphics{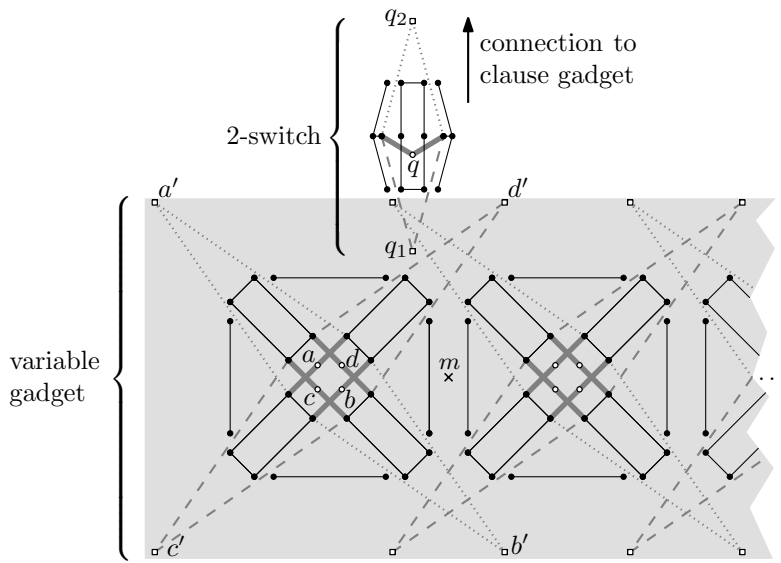}
    \caption{Edges in the variable gadget: immobile (solid
      black) and mobile (thick solid gray).  The predestined
      positions of mobile edges either correspond to \emph{true}
      (dashed gray) or to \emph{false} (dotted gray).}
    \label{fig:variable}
  \end{figure}
  The gadget consists
  of a horizontal chain of a certain number of roughly square
  \emph{blocks}.  Each block consists of 28 vertices (four of which
  are mobile) and 28 edges.
  In \figurename~\ref{fig:variable} the four mobile vertices of
  the leftmost block are labeled in clockwise order~$a$, $d$, $b$, 
  and~$c$.  Note that the gray edges incident to~$a$ and~$b$ intersect
  those incident to~$c$ and~$d$.  Thus either both~$a$ and~$b$ or
  both~$c$ and~$d$ must be moved to untangle the block.  Each mobile vertex 
  $w \in \{a,b,c,d\}$ can move into exactly one position~$w'$ (up to
  small perturbations).  The resulting incident edges are drawn by dotted
  and dashed gray line segments, respectively.  Note that
  neighboring blocks in the chain are placed such that the only way to
  untangle them simultaneously is to move \emph{corresponding} pairs
  of vertices and edges.  Thus either all blocks of a variable gadget
  use the dashed line segments or all use the dotted line segments.  These
  two ways to untangle a variable gadget correspond to the values
  \emph{true} and \emph{false} of the variable, respectively.

  Let~$\clau$ be the numbers of clauses of~$\varphi$.
  For each of the~$3\clau$ literals in~$\varphi$ we connect the gadget of
  the corresponding variable to the gadget of the clause that contains
  the literal.  Each block of each variable gadget is connected to a
  specific clause gadget above or below the variable gadget, thus
  there are~$3\clau$ blocks in total.  Each connection is realized by a
  part of~$G_\varphi$ that we call a \emph{2-switch}.  A 2-switch
  consists of 15 vertices and 14 edges.  The mobile vertex~$q$ of the
  2-switch in \figurename~\ref{fig:variable} is incident to two 
  thick gray edges that intersect two immobile edges of the 2-switch.
  Thus~$q$ must move.  There are (up to small perturbations) two
  possible positions, namely~$q_1$ and~$q_2$, see
  \figurename~\ref{fig:variable}. 

  The 2-switch in \figurename~\ref{fig:variable} corresponds to a positive
  literal.  For negated literals the switch must be mirrored either at
  the vertical or at the horizontal line that runs through the point~$m$.  
  Note that a switch can be stretched vertically in order to reach
  the right clause gadget.  Further note that if a literal is
  \emph{false}, the mobile vertex of the corresponding 2-switch must
  move away from the variable gadget and towards the clause gadget to
  which the 2-switch belongs.  In that case we say that the 2-switch
  \emph{transmits pressure}.  

  A clause gadget consists of three vertical 2-switches and two
  horizontal \emph{3-switches}.  A 3-switch consists of 23 vertices
  and 18 edges plus a small ``inner'' 2-switch, see the shaded
  area in \figurename~\ref{fig:clause}.  
  \begin{figure}
    \centering %
    \includegraphics{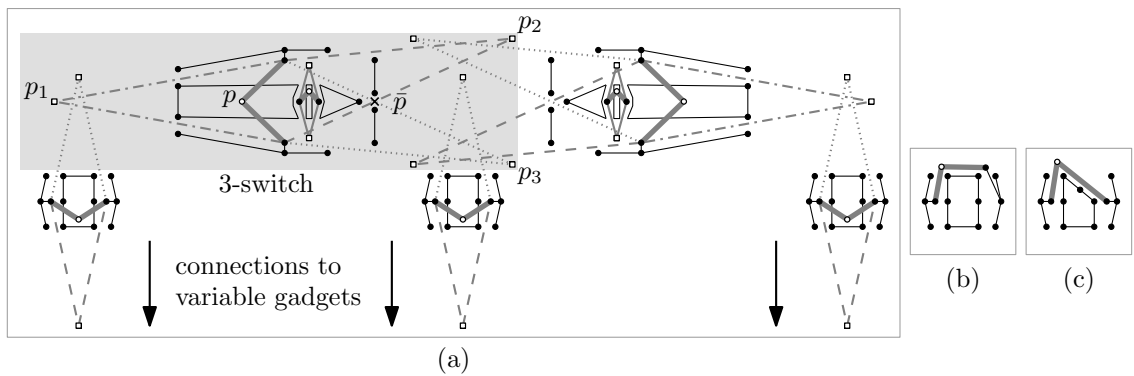}
    \caption{(a) A clause gadget consists of three big 2-switches
      (drawn vertically) and two 3-switches (drawn horizontally; one
      is shaded).  Each 3-switch contains another small 2-switch.
      Note that not all immobile vertices are marked. (b) \& (c) Two
      ways in which originally immobile vertices can move to avoid a
      crossing if~$\varphi$ is not satisfiable.}
    \label{fig:clause}
  \end{figure}
  Independently from the other, each
  of the two 3-switches can be stretched horizontally in order to
  reach vertically above the variable gadget to which it connects via
  a 2-switch.  The mobile vertex~$p$ of the left 3-switch in
  \figurename~\ref{fig:clause} is incident to two thick gray edges
  that intersect two immobile edges of the 3-switch.  Thus~$p$ must move.
  There are (again up to small perturbations) three possible positions, 
  namely~$p_1$, $p_2$, and~$p_3$.  Note that we need the inner 2-switch,
  otherwise there would be a forth undesired position for moving~$p$,
  namely the one labeled~$\bar{p}$ in \figurename~\ref{fig:clause}.  
  By construction, a clause gadget can be made plane by only moving
  the mobile vertices of all switches if and only if at most two of
  the three big 2-switches transmit pressure, that is, if at least one of
  the literals in the clause is \emph{true}.  

  The graph~$G_\varphi$ that we have now constructed has~$O(\clau)$
  vertices, $O(\clau)$ edges, and $X = 26\clau$ crossings; $4 \cdot 3\clau$ in
  blocks and $2 \cdot 7\clau$ in switches.  Recall that any mobile vertex
  is incident to two edges that each cross another edge.  Thus a
  mobile vertex corresponds to a pair of crossings. 
  By moving a mobile vertex to any of its predestined positions, 
  the corresponding pair of crossings  
  disappears.  If~$\varphi$ is satisfiable, $G_\varphi$ can be made
  plane by moving $K=X/2$ mobile vertices since no new crossings are
  introduced.  If~$\varphi$ is not satisfiable, there is at least one
  pair of crossings that cannot be eliminated by moving the
  corresponding mobile vertex alone since all its predestined
  positions are blocked.  Thus at least \emph{two} vertices must be
  moved to eliminate that pair of crossings---and still all the other
  $K-1$ pairs of crossings must be eliminated by moving at least one
  vertex per pair, totaling
  in at least $K+1$ moves.  Thus~$\varphi$ is satisfiable if and only
  if~$G_\varphi$ can be made plane by moving exactly~$K$ (mobile)
  vertices.   

  Recall that~$G_\varphi$ consists of~$O(\clau)$ vertices and edges.  We
  construct~$\delta_\varphi$ step by step, starting with the vertices
  of the variable gadgets and then treating the clauses from innermost
  to outermost.  In order for the 2- and 3-switches to reach far
  enough, note that each desired position of a mobile vertex is
  determined by two pairs of immobile vertices.  By making the
  distances of the two vertex pairs (polynomially) small, the desired
  position can be confined to a region that is small enough to force
  the mobile vertex of the next switch into one of its remaining
  positions.  Now it is clear that it is possible to place vertices at
  coordinates whose representation has size polynomial in the
  length~$L$ of a binary encoding of~$\varphi$.  This implies that our
  reduction is polynomial in~$L$.
\end{proof} 

\begin{remark}
  Our proof can be slightly modified to show that the problem is also
  hard if we are additionally given an axis-parallel rectangle that
  contains the initial graph drawing, and each move is constrained to
  that rectangle---in other words Tantalo's version of the planarity
  game.  In the proof we simply compute from the given planar 3-SAT
  formula a rectangle that is large enough to accommodate not only the
  initial drawing, but also the plane drawing that we get in case the
  formula has a satisfying truth assignment.  Note that this rectangle
  is barely larger than the smallest axis-parallel rectangle that
  contains all vertices of our initial graph drawing.
\end{remark}

We now consider the approximability of \textsc{MinShiftedVertices}.
Since $\shift(G,\delta)=0$ for plane drawings, we cannot use the usual
definition of an approximation factor unless we slightly modify our
objective function.  Let $\shift'(G,\delta)=\shift(G,\delta)+1$ and call the
resulting decision problem \textsc{MinShiftedVertices}$'$.  Now we can
modify the above reduction to get a non-approximability result.

\begin{theorem} \label{thm:non-approx-mmv}
  For any constant real $\eps \in (0,1]$ there is no polynomial-time
  $n^{1-\eps}$-approximation algorithm for
  \textsc{MinShiftedVertices}$'$ unless ${\cal P} = \cal NP$.
\end{theorem}

\begin{proof}
  Let~$n_\varphi$ be the number of vertices of the graph~$G_\varphi$
  with drawing~$\delta_\varphi$ that we constructed above.  We add
  to~$G_\varphi$ for each edge~$e$ $n_\varphi^{(3-\eps)/\eps}$ copies,
  half of them on each side of~$e$, in the close vicinity of~$e$.  
  If one of the endpoints of~$e$ is a mobile vertex, then all copies
  are incident to that vertex.  In the following we detail where to
  place the other (new) endpoints of these edges.

  We go through each immobile vertex~$v$ of~$G_\varphi$.  Let
  $\deg_\varphi v$ be the degree of~$v$ in~$G_\varphi$.  Note that $1
  \le \deg_\varphi v \le 3$.  If $\deg_\varphi v = 1$, we place the
  endpoints of the copies of the edge~$e$ incident to~$v$ on the two
  rays that are orthogonal to~$e$ in~$v$.  On each ray we place half
  of the endpoints and connect them by new edges along the ray,
  starting with~$v$, see vertex~$v_1$ in
  \figurename~\ref{fig:variable-modified}.

  Otherwise, if $\deg_\varphi v > 1$, let~$e$, $e'$ be two edges that
  are incident to~$v$ and consecutive in the circular ordering
  around~$v$.  Now we add half of the endpoints of~$e$ and~$e'$
  on a ray between~$e$ and~$e'$ emanating from~$v$, in the same manner
  as above.  The position of the ray depends on whether both~$e$
  and~$e'$ are immobile or one of them is mobile.
  (Being immobile,
  vertex~$v$ is incident to at most one mobile edge.)  In the first
  case we place the new vertices on the angular bisector of~$e$
  and~$e'$, see vertex~$v_2$ in
  \figurename~\ref{fig:variable-modified}.  In the second case where
  one of the edges, say~$e$, is mobile, note that the original and all
  predestined positions of~$e$ lie in an open halfplane bounded by a
  line~$\ell$ through~$v$.  So we place the new vertices on~$\ell$,
  half on each side of~$v$, see vertex~$v_3$ in
  \figurename~\ref{fig:variable-modified}.

  Let~$G$ be the resulting graph, $\delta$ its drawing, and $n \le
  \big(3/2 \cdot n_\varphi^{(3-\eps)/\eps}+1\big) \cdot n_\varphi$ the number of
  vertices of~$G$.  Note that~$\varphi$ is satisfiable if and only if
  $\shift'(G,\delta)= \shift'(G_\varphi,\delta_\varphi) = K+1$.
  Otherwise, in the original graph~$G_\varphi$ at least one immobile
  vertex has to move.  This vertex either is incident to a mobile edge
  or it is not, see \figurename~\ref{fig:clause}(b) and~(c),
  respectively.  In the new graph~$G$, which contains~$G_\varphi$,
  also at least one (original) immobile vertex~$v$ has to move.
  If~$v$ is not incident to a mobile edge, in order to make space,
  all new vertices in the vicinity of~$v$ have to move, too.  If~$v$
  is incident to a mobile edge, a new vertex in the vicinity of~$v$
  has to move only if it is incident to a new copy of the mobile edge.
  That is, in both cases, at least
  $n_\varphi^{(3-\eps)/\eps}$ vertices have to move.  In other words,
  $\shift'(G,\delta) \ge K+2+n_\varphi^{(3-\eps)/\eps}$.  Note that~$G$
  can be constructed in polynomial time since we have assumed~$\eps$
  to be a constant. 

  \begin{figure}
    \centering %
    \includegraphics[scale=.9]{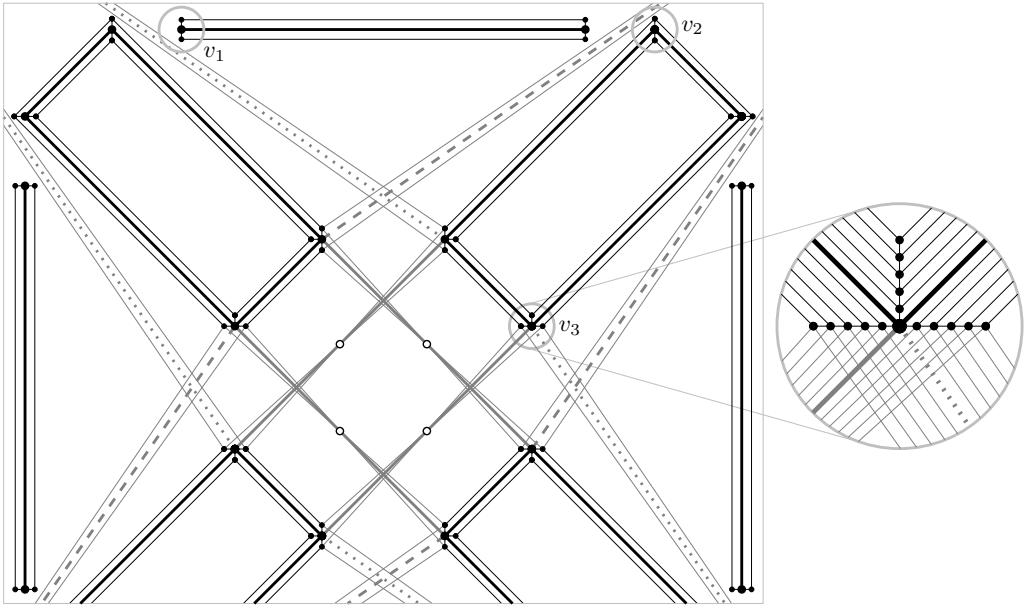}
    \caption{Clipping of the modified variable gadget for the proof of
      Theorem~\ref{thm:non-approx-mmv}.  The old vertices and edges
      are drawn thicker than the new ones.  Each old edge has
      $n_\varphi^{(3-\eps)/\eps}$ new copies.}
    \label{fig:variable-modified}
  \end{figure}

  Suppose there was a polynomial-time $n^{1-\eps}$-approximation
  algorithm~$\cal A$ for \textsc{MinShiftedVertices}$'$.  We can bound
  its approximation factor by
  $n^{1-\eps} \le
  \big((3/2 \cdot n_\varphi^{(3-\eps)/\eps}+1) \cdot n_\varphi\big)^{1-\eps}
  \le \big(2 n_\varphi^{(3-\eps)/\eps} \cdot n_\varphi\big)^{1-\eps}
  = 2^{1-\eps} n_\varphi^{(3-3\eps)/\eps}
  \le 2n_\varphi^{(3-3\eps)/\eps}$.
  Now let~$M$ be the number of moves that~$\cal A$ needs to
  untangle~$\delta$.  If~$\varphi$ is satisfiable, then $M \le 
  \shift'(G,\delta) \cdot n^{1-\eps} = (K+1) \cdot n^{1-\eps} \le
  (n_\varphi+1) \cdot 2n_\varphi^{(3-3\eps)/\eps}
  = 2n_\varphi^{(3-2\eps)/\eps} + 
  O\big(n_\varphi^{(3-3\eps)/\eps}\big)$.  On the other hand, if
  $\varphi$ is unsatisfiable, then $M \ge \shift'(G,\delta) >
  n_\varphi^{(3-\eps)/\eps}$.  Since we can assume that~$n_\varphi$ is
  sufficiently large, the result of algorithm~$\cal A$ (that is, the
  number~$M$) tells us whether~$\varphi$ is satisfiable.  So either
  our assumption concerning the existence of~$\cal A$ is wrong, or we
  have shown the NP-hard problem \textsc{Planar3SAT} to lie in~$\cal
  P$, which in turn would mean that ${\cal P} = \cal NP$.
\end{proof}

We now state a hardness result that establishes a connection between
\textsc{MinShiftedVertices} and the well-known graph-drawing problem
\textsc{1BendPointSetEmbeddability}.  We define the problem
\textsc{1BendPointSetEmbeddabilityWithCorrespondence} as follows.
Given a planar graph $G=(V,E)$, a set~$S$ of points in the plane with
rational coordinates and a one-to-one correspondence~$\zeta$
between~$V$ and~$S$, decide whether~$\zeta$ can be
extended to a plane 1-bend drawing of~$G$, that is, whether~$G$ has a
plane drawing~$\delta$ such that $\delta(v) = \zeta(v)$ for all $v \in
V$ and such that~$\delta$ maps each edge of~$G$ to a 1-bend polygonal chain.

\begin{theorem} \label{thm:hardness-one-bend} %
  \textsc{1BendPointSetEmbeddabilityWithCorrespondence} is NP-hard.
\end{theorem}

\begin{proof}
  The proof uses nearly the same gadgets as in the proof of
  Theorem~\ref{thm:hardness-mmv}: set~$G'_\varphi$ to a copy
  of~$G_\varphi$ where each length-2 path $(u,v,w)$ containing a
  mobile vertex~$v$ is replaced by the edge~$uw$.  We refer to
  this type of edges as \emph{new} edges.  The vertices
  of~$G'_\varphi$ are mapped to the corresponding vertices
  in~$\delta_\varphi$.  We claim that~$G'_\varphi$ has a 1-bend
  drawing if and only if the given planar-3SAT formula~$\varphi$ is
  satisfiable.

  In order to see that the claim holds, note the two differences to
  the proof of Theorem~\ref{thm:hardness-mmv}.  First, in 
  \textsc{1BendPointSetEmbeddabilityWithCorrespondence} all vertices
  \emph{are} fixed.  This makes it even easier to argue correctness.
  Second, any edge can bend, not only new edges, which are meant to
  bend.  Due to the fact that vertices cannot move, however, all
  groups of edges that are meant to be obstacles will remain obstacles
  to the bending of the new edges.  The only way to embed the new
  edges is to route them around the obstacles exactly as in
  Figs.~\ref{fig:variable} and~\ref{fig:clause}(a).
\end{proof}

Now suppose that we already know that~$G$ has a plane drawing with at most
one bend per edge.  Then it is natural to ask for a drawing with as
few bends as possible.  Let~$\beta(G)$ be 1 plus the minimum number of
bends over all plane 1-bend drawings of~$G$.  The following
corollary shows that it is hard to approximate $\beta(G)$ efficiently.

\begin{corollary} 
  \label{cor:non-approx-one-bend}
  Given a planar graph $G=(V,E)$, a set $S \subset \mathbb{Q}^2$, a
  one-to-one correspondence~$\zeta$ between~$V$ and~$S$ that can be
  extended to a plane 1-bend drawing of~$G$, and a constant $\eps \in
  (0,1]$, it is NP-hard to approximate $\beta(G)$ within a factor 
  of~$n^{1-\eps}$.
\end{corollary}

\begin{proof}
  We slightly change the clause gadget in the proof of
  Theorem~\ref{thm:hardness-one-bend}.  Apart from the three vertical
  2-switches, the clause gadget now consists of two 4-switches and of
  two stacks of~$s$ edges each, see \figurename~\ref{fig:non-approx}.
  Let~$G''_\varphi$ be the resulting graph, which depends on the given
  planar 3SAT formula~$\varphi$.  The 4-switches make sure
  that~$G''_\varphi$ always has a drawing with at most one bend per edge. 
  Each stack is placed in the vicinity of a 4-switch such that all
  stack edges have to bend if the central switch edge is forced to
  bend into the direction of the stack.  (In
  \figurename~\ref{fig:non-approx}, the central switch edges in the left and
  right 4-switch are labeled~$e_C$ and~$e_C'$, respectively.)
  If~$\varphi$ is not satisfiable, at least one clause evaluates to
  \emph{false} and in the corresponding gadget all~$s$ edges in the
  left or all~$s$ edges in the right stack need to bend.

  The number~$s$ of edges per stack can be set
  to~${n'_{\varphi}}^{(3-\eps)/\eps}$, where~$n'_{\varphi}$ is the
  number of vertices of the graph~$G'_{\varphi}$ defined in the proof
  of Theorem~\ref{thm:hardness-one-bend}.  Then, the remaining
  calculations for proving hardness of approximation are similar to
  those in the proof of Theorem~\ref{thm:non-approx-mmv}.
\end{proof}

\begin{figure}
  \centering
  \includegraphics[width=\textwidth]{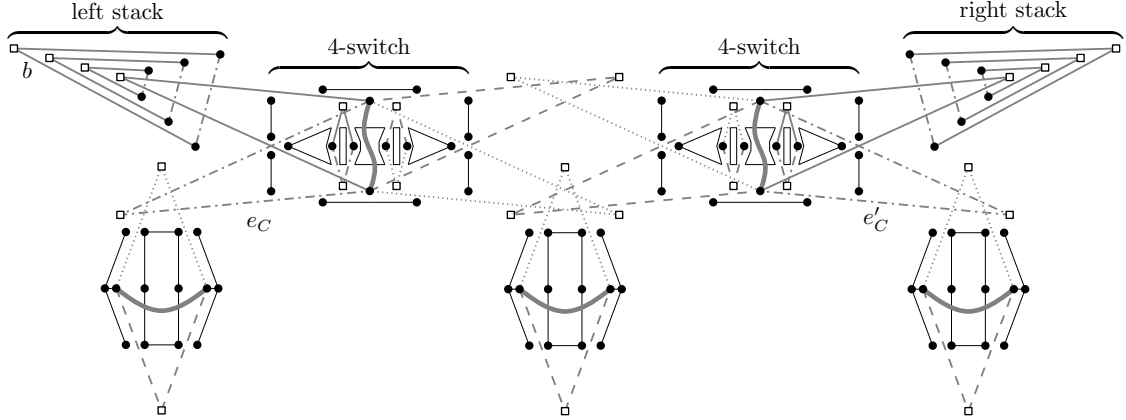}
  \caption{Gadget of clause~$C$ adapted for the proof of
    Corollary~\ref{cor:non-approx-one-bend}.  Edges~$e_C$ and~$e'_C$
    each belong to a 4-switch, that is, they can be drawn in \emph{four}
    combinatorially different ways (drawn in gray; solid vs.\
    dashed-dotted vs.\ dotted vs.\ dashed).  Note that not all
    vertices are marked.}
  \label{fig:non-approx}
\end{figure}

We do not know whether
\textsc{1BendPointSetEmbeddabilityWithCorrespondence} or
\textsc{MinShiftedVertices} lie in~$\cal NP$, but it is not hard to
show the following.

\begin{theorem}
  \label{thm:1bend-pspace}
  \textsc{1BendPointSetEmbeddabilityWithCorrespondence} is in \pspace.
\end{theorem}

\begin{proof}
  Let $G=(V,E)$ be a planar graph, $S$ a set of~$n$ points in the
  plane with rational coordinates, and~$\zeta$ a one-to-one
  correspondence between~$V$ and~$S$.  Any 1-bend drawing of~$G$ that
  extends~$\zeta$ is uniquely determined by choosing, for each
  edge~$e$, the position $(x_e,y_e)$ of the bend~$b_e$ of~$e$.  (If an
  edge~$uv$ is to be drawn without bend, any point in the relative
  interior of the line segment connecting~$\zeta(u)$ and~$\zeta(v)$
  can be chosen.)  Thus, the set of all plane 1-bend drawings of~$G$
  that extend~$\zeta$ can be represented by a subset of~$\R^{2|E|}$.  The 
  bend~$b_e$ splits (the drawing of) the edge~$e$ into two relative
  open line segments to which we refer as \emph{half-edges}.

  In order to decide the existence of a plane 1-bend drawing,
  we specify a predicate in polynomial inequalities with integer
  coefficients and with variables in the set ${\cal E} = \{x_e, y_e
  \mid e \in E\}$.  We do this by first expressing the condition that
  no two half-edges with distinct endpoints may intersect.
  Given four distinct points~$A$, $B$, $C$, and~$D$, the requirement
  that points~$C$ and~$D$ lie in different half-planes determined by
  the line through~$A$ and~$B$ can be expressed by an inequality
  $P(A,B,C,D)<0$, where~$P$ is a degree-4 polynomial with integer
  coefficients and with variables representing the coordinates of the
  four points~\cite{km-igs-94}.  The requirement that the line
  segments~$AB$ and~$CD$ are disjoint is described by the disjunction
  $(P(A,B,C,D)>0) \vee (P(C,D,A,B)>0)$.

  Second, we add conditions that guarantee that no bend~$b_e$
  coincides with a point in~$S$, that all bends are distinct, and that
  no two half-edges overlap if they share an endpoint.  All these
  conditions can also be described as Boolean combinations of
  polynomial inequalities with integer coefficients and with variables
  from~$\cal E$.  As a consequence, deciding whether~$\zeta$ extends
  to a 1-bend drawing of~$G$ recasts into deciding the non-emptiness
  of a set in~$\R^{2|E|}$ defined by a predicate whose atomic formulas
  are polynomial inequalities with integer coefficients, a problem
  that is in \pspace~\cite{c-sagcp-88,r-ccgfo-92}.
\end{proof}

For \textsc{MinShiftedVertices} and \textsc{MaxFixedVertices} an
additional trick is needed. 

\begin{proposition} 
  \label{pro:vertexmove-pspace}
  \textsc{MinShiftedVertices} and \textsc{MaxFixedVertices} are in
  \pspace.
\end{proposition}

\begin{proof}
  Obviously, both problems have the same optimal solutions, so it is
  enough to treat one of them, say \textsc{MinShiftedVertices}.  We
  build on the formulation sketched in the proof of
  Theorem~\ref{thm:1bend-pspace}.  Additionally, we introduce a binary
  variable~$z_v$ for each vertex~$v$ that encodes whether we move
  vertex~$v$ ($z_v = 1$) or not ($z_v = 0$).  In order to
  restrict~$z_v$ to these two values, we introduce
  the quadratic equation $z_v (z_v - 1) = 0$.  The
  $x$-coordinate of vertex~$v$ in the plane target drawing can then be
  described by $(1-z_v) X_v + z_v x_v$, where~$X_v$ is the original
  $x$-coordinate of~$v$ and~$x_v$ is the $x$-coordinate of~$v$ after a
  possible movement.  The $y$-coordinate of~$v$ is treated
  analogously.  The intersection of edges can be expressed as in the
  proof of Theorem~\ref{thm:1bend-pspace}.  To bound the number of
  moved vertices by~$K$, we introduce the inequality $\sum_{v \in V}
  z_v \le K$.
\end{proof}

\section{Planar Graphs: Lower Bound} \label{sec:planar:lowerbound}

Any drawing of a planar graph with $n \ge 3$ vertices, other
than~$K_3$ or~$K_4$, can be untangled while fixing at least three 
vertices~\cite{v-ocpg-08}. In this section, we give an algorithm
proving that
\[\fix(G) \ge f(n) = \sqrt{\frac{(\log n) -1}{\log \log
n}}\]
for any planar graph~$G$ with $n \ge 4$ vertices. Note that~$f$
actually grows, albeit \emph{very} slowly: $f(n) > 3$ only for some
$n \approx 6 \cdot 10^{15}$.  Partially building on our algorithm, Bose et
al.~\cite{bdhlmw-pbugp-08} showed that $\fix(G) \ge \sqrt[4]{(n+1)/2}$, a
bound greater than~$3$ for $n>161$. 

We first give some definitions (Section~\ref{sub:definitions}) and
sketch the basic idea of our algorithm (Section~\ref{sub:basic}).
Then we describe our algorithm
(Section~\ref{sub:description:algorithm}) and prove its correctness
(Section~\ref{sub:main:theorem}).  The bound $\fix(G) \ge f(n)$
depends on finding a plane embedding of~$G$ that contains a long
simple path with an additional property.  We show how to find such an
embedding in Section~\ref{sub:finding:a:path}.

\subsection{Definitions and notation}
\label{sub:definitions}

\renewcommand{\textfraction}{0.01} \renewcommand{\topfraction}{0.99}
\renewcommand{\bottomfraction}{0.99}

Recall that a \emph{plane embedding} of a planar graph is given by the
circular order of the edges around each vertex and by the choice of
the outer face.  A plane embedding of a planar graph can be computed
in linear time \cite{ht-ept-74}.  If~$G$ is triangulated, a plane
embedding of~$G$ is determined by the choice of the outer face.
Further recall that an edge of a graph is called a \emph{chord} with
respect to a path (or cycle)~$\Pi$ if the edge does not lie on~$\Pi$
but both its endpoints are vertices of~$\Pi$.

For a point $p \in \R^2$, let~$x(p)$ and~$y(p)$ be the $x$- and
$y$-coordinates of~$p$, respectively.  We say that~$p$ lies
\emph{vertically below} $q \in \R^2$ if $x(p)=x(q)$ and $y(p) \le
y(q)$.  For a polygonal path $\Pi=v_1,\dots,v_k$, we denote by
$V_\Pi=\{v_1,\dots,v_k\}$ the set of vertices of~$\Pi$ and by $E_\Pi =
\{ v_1v_2, \dots, v_{k-1}v_k \}$ the set of edges of~$\Pi$.  We call a
polygonal path $\Pi=v_1,\dots,v_k$ \emph{$x$-monotone} if $x(v_1) <
\dots < x(v_k)$.  In addition, we say that a point $p \in \R^2$ lies
\emph{below} an $x$-monotone path~$\Pi$ if~$p$ lies vertically below a
point~$p'$ (not necessarily a vertex!) on~$\Pi$.  Analogously, a line
segment $\overline{pq}$ lies below~$\Pi$ if every point $r \in
\overline{pq}$ lies below~$\Pi$.  We do not always strictly
distinguish between a vertex~$v$ of~$G$ and the point~$\delta(v)$ to
which this vertex is mapped in a particular drawing~$\delta$ of~$G$.
Similarly, we write~$vw$ both for the edge of~$G$ and the
straight-line segment connecting~$\delta(v)$ with~$\delta(w)$.

\subsection{Basic idea}
\label{sub:basic}

Note that in order to establish a lower bound on~$\fix(G)$, we can
assume that the given graph~$G$ is triangulated.  Otherwise we can
triangulate~$G$ arbitrarily (by fixing an embedding of~$G$ and adding
edges until all faces are 3-cycles) and work with the resulting
triangulated planar graph.  A plane drawing of the latter
yields a plane drawing of~$G$.  So let~$G$ be a triangulated planar
graph, and let~$\delta_0$ be any drawing of~$G$.
It will also be convenient to assume
that in the given drawing $\delta_0$, the vertices of $G$
are mapped to points with pairwise distinct $x$-coordinates.
By slightly rotating the drawing $\delta_0$ we can always achieve this.

The basic idea of our algorithm is to find a plane embedding~$\beta$
of~$G$ such that there exists a long simple path~$\Pi$ connecting two
vertices~$s$ and~$t$ of the outer triangle~$stu$ with the property
that all chords of~$\Pi$ lie on one side of~$\Pi$ (with respect
to~$\beta$) and~$u$ lies on the other.  For an example of such an
embedding~$\beta$, see \figurename~\ref{figure:steps:basic:idea}(b).  We
describe how to find~$\beta$ and~$\Pi$ depending on the maximum degree
and the diameter of~$G$ in Section~\ref{sub:finding:a:path}.  For
the time being, we assume they are given.  Now our goal is to
produce a drawing of~$G$ according to the embedding~$\beta$ and at the
same time keep many of the vertices of~$\Pi$ at their positions
in~$\delta_0$. Having all chords on one side is the crucial property
of~$\Pi$ that we use to achieve this.  We allow ourselves to move all other
vertices of~$G$ to any location we like, a process we will occasionally
refer to as \emph{drawing} certain subgraphs of~$G$. This gives us a lower bound
on $\fix(G,\delta)$ in terms of the number~$l$ of vertices of~$\Pi$.
Our method is illustrated in \figurename~\ref{figure:steps:basic:idea}; we
give the details in the next subsection.

\begin{figure}[t]
  \centering
  \includegraphics{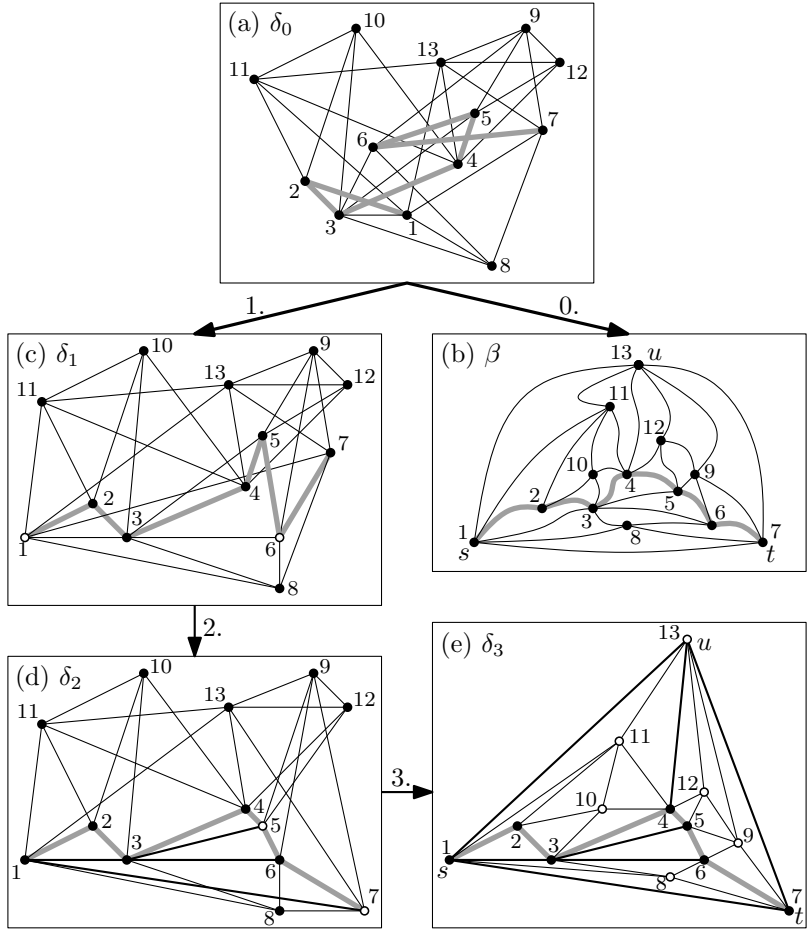}
  \caption{An example run of our algorithm: (a)~input: the given
    non-plane drawing~$\delta_0$ of a triangulated planar graph~$G$.
    (b)~Plane embedding~$\beta$ of~$G$ with path~$\Pi$ (drawn in gray)
    that connects two vertices on the outer face.  To untangle~$\delta_0$
    we first make~$\Pi$ $x$-monotone~(c), then we bring all
    chords (bold segments) to one side of~$\Pi$~(d), move~$u$ to a
    position on the other side of~$\Pi$ where~$u$ sees all vertices
    in~$V_\Pi$, and finally move  
    the vertices in $V \setminus (V_\Pi \cup \{u\})$ to suitable
    positions within the faces bounded by the bold gray and black
    edges~(e).  Vertices that move from~$\delta_{i-1}$ to~$\delta_i$
    are marked by circles; those that do not move are marked by black
    disks.}
  \label{figure:steps:basic:idea}
\end{figure}

\subsection{Description of the algorithm}
\label{sub:description:algorithm}

Let~$C$ denote the set of chords
of~$\Pi$.  We assume that these chords lie to the right of~$\Pi$ 
when we traverse this path from~$s$ to~$t$ in
the embedding~$\beta$.  (Note that ``below'' is not defined in an
embedding.)  Let \Vbot\ denote the set of vertices of~$G$ that lie to
the right of~$\Pi$ in~$\beta$ and let $\Vtop = V \setminus (V_\Pi \cup
\Vbot)$.  Note that~$u$ lies in \Vtop.  Let~$I$ be a subset of the
vertices of~$\Pi$ such that no two vertices in~$I$ are connected by a
chord of~$\Pi$.  We will choose~$I$ such that $|I| \ge
(l+1)/2$, and our method tries to fix many of the vertices in~$I$.

In step~1 of our algorithm we move some of the vertices in~$V_\Pi$ from the
position they have in~$\delta_0$ to new positions such that
the resulting ordering of the vertices in~$V_\Pi$ according to
increasing $x$-coordinates is the same as the ordering along~$\Pi$ in~$\beta$.
This yields a new (usually non-plane) drawing~$\delta_1$ of~$G$ that
maps~$\Pi$ on an $x$-monotone polygonal path~$\Pi_1$.  By
Theorem~\ref{erdos-szekeres} we can choose~$\delta_1$ such that 
at least~$\sqrt{|I|}$ of the vertices in~$I$
remain fixed.  Let $F \subseteq I \subseteq V_\Pi$ be the set of the
fixed vertices.  Note that $\delta_1(v)=\delta_0(v)$ for all $v \in V
\setminus V_\Pi$, see \figurename~\ref{figure:steps:basic:idea}(c).

Once we have constructed~$\Pi_1$, we have to find suitable positions
for the vertices in $\Vtop \cup \Vbot$.  This is simple for the
vertices in \Vtop: if we move vertex~$u$, which lies on the outer
face, far enough above~$\Pi_1$, then the polygon~$P_1$ bounded
by~$\Pi_1$ and by the edges~$us$ and~$ut$ will be \emph{star-shaped}. 
Recall that a polygon~$P$ is called \emph{star-shaped} if the interior
of its kernel is non-empty, and the \emph{kernel} of a
clockwise-oriented polygon~$P$ is the intersection of the right
half-planes induced by the edges of~$P$.  Now if~$P_1$ is star-shaped,
we have fulfilled one of the assumptions of the following result of
Hong and Nagamochi~\cite{hn-cdgnc-08} for drawing \emph{triconnected}
graphs, that is, graphs that cannot be disconnected by removing two
vertices.  We will use their result in order to draw
into~$P_1$ the subgraph~\Gtop\ of~$G$ induced by $\Vtop \cup V_\Pi$
\emph{excluding} the chords in~$C$.

\begin{theorem}[\cite{hn-cdgnc-08}]
  \label{thm:hong-nagamochi}
  Given a triconnected plane graph~$H$, every drawing~$\delta^*$ of
  the outer facial cycle of~$H$ on a star-shaped polygon~$P$ can be
  extended in linear time to a plane drawing of~$H$.
\end{theorem}

Observe, however, that \Gtop\ is \emph{not} necessarily triconnected:
vertex~$u$ may be adjacent to vertices on~$\Pi$ other than~$s$
and~$t$.  In order to fix this, we split~\Gtop\ into smaller units
along the edges incident to~$u$.  Let $(s=)w_1,w_2,\dots,w_l(=t)$ be
the sequence of vertices of~$\Pi$.  Let~$(i,k)$ be a pair of
integers such that $1 \le i < k \le l$, vertices~$w_i$ and~$w_k$ are
adjacent to~$u$ and any vertex~$w_j$ with $i<j<k$ is not adjacent
to~$u$.  Consider the subgraph of~\Gtop\ induced by the vertices that
lie (with respect to~$\beta$) inside of or on the cycle
$u,w_i,w_{i+1},\dots,w_j$.  In the following we convince ourselves
that this subgraph is actually triconnected.
Let~\Htop\ be the family of all such subgraphs.

Recall that a planar graph~$H$ is called a
\emph{rooted triangulation} \cite{a-grtwr-96} if in every plane
drawing of~$H$ there exists at most one facial cycle with more than
three vertices.  According to Avis \cite{a-grtwr-96},
the following lemma is well known.  

\begin{lemma}[\cite{a-grtwr-96}]
  \label{lemma:wheellike:implies:triconnected}
  A rooted triangulation is triconnected if and only if no facial
  cycle has a chord.
\end{lemma}

Now it is clear that we can apply Theorem~\ref{thm:hong-nagamochi} to
draw each of the subgraphs in~\Htop.
By the placement of~$u$, each drawing region is star-shaped,
and by construction, each subgraph is chordless and thus triconnected.
However, to draw the graph \Gbot\ induced by $\Vbot \cup V_\Pi$
(\emph{including} the chords in~$C$), we must work a little harder.

In step~2 of our algorithm we once more change the drawing of~$\Pi$.
Let $V^* = V_{\Pi} \setminus I$.  Note that every chord of~$\Pi$ has at
least one of its endpoints in~$V^*$.  Now we go through the vertices
in~$V^*$ in a certain order, moving each vertex vertically down as far
as necessary (see vertices~5 and~7 in
\figurename~\ref{figure:steps:basic:idea}(d)) to achieve two goals: (a)~all
chords in~$C$ move below the resulting polygonal path~$\Pi_2$, and
(b)~the faces bounded by~$\Pi_2$, the edge~$st$, and the chords become
star-shaped polygons.  This defines a new drawing~$\delta_2$, which
leaves all vertices in~$F$ and all vertices in $V \setminus V_\Pi$
fixed.

In step~3 we use the fact that~$\Pi_2$ is still $x$-monotone.  This allows us
to move vertex~$u$ to a location above~$\Pi_2$ where it can see every
vertex of~$\Pi_2$.  Now~$\Pi_2$, the edges of type~$uw_i$ (with
$1<i<l$) and the chords in~$C$ partition the triangle~$ust$ into
star-shaped polygons with the property that the subgraphs of~$G$ that
have to be drawn into these polygons are all rooted triangulations,
and thus triconnected.  This means that we can apply
Theorem~\ref{thm:hong-nagamochi} to each of them.  The result is our
final---and plane---drawing~$\delta_3$ of~$G$, see
\figurename~\ref{figure:steps:basic:idea}(e).

\subsection{Correctness of the algorithm}
\label{sub:main:theorem}

We now show that our algorithm indeed produces a plane drawing where
many of the vertices on the chosen path~$\Pi$ are fixed. To this end,
recall that $F \subseteq I$ is the set of vertices in~$V_\Pi$ we fixed
in step~1, that is, in the construction of the $x$-monotone polygonal
path~$\Pi_1$.  Our goal is to fix the vertices in~$F$ when we
construct~$\Pi_2$, which also is an $x$-monotone polygonal path but
has two additional properties: (a)~all chords in~$C$ lie below~$\Pi_2$
and (b)~the faces induced by~$\Pi$, $w_1 w_l$, and the chords in~$C$
are star-shaped polygons.  The following lemmas form the basis for the
proof of the main theorem of this section
(Theorem~\ref{theorem:main}), which shows that this can be achieved.

\begin{lemma}
  \label{lemma:moving:downwards:is:okay:for:kernel}
  Let $\Pi = v_1,\dots,v_k$ be an $x$-monotone polygonal path such
  that (i)~the segment~$v_1 v_k$ lies below~$\Pi$ and (ii)~the
  polygon~$P$ bounded by~$\Pi$ and~$v_1 v_k$ is star-shaped.  
  Let~$v_k'$ be any point vertically below~$v_k$.  Then the polygon $P'=
  v_1,\dots,v_{k-1},v_k'$ is also star-shaped.
\end{lemma}

\begin{proof}
  Only two edges change when we move vertex~$v_k$ to its new
  position~$v_k'$, namely~$v_1 v_k$ and~$v_{k-1} v_k$.  Consider the remaining 
  $k-2$ edges that do not change and let~$K$ be the intersection of
  the corresponding right half-planes.  Since the~$k-2$ edges form an
  $x$-monotone path, $K$ is not bounded.  Let~$q$ be a point in the
  interior of the kernel of~$P$.  Then~$q$ lies in the interior of~$K$
  and, moreover, every point that is vertically below~$q$ also lies in
  the interior of~$K$. Let~$q'$ be a point vertically below~$q$ and
  sufficiently close to the edge~$v_1 v_k'$. 
  Then~$q'$ lies in the interior of the kernel of~$P'$, and
  therefore~$P'$ is star-shaped by definition. 
\end{proof}

\begin{lemma}
  \label{lemma:moving:downwards:is:okay:for:chords}
  Let $\Pi=v_1,\dots,v_k$ be an $x$-monotone polygonal path, and
  let~$D$ be a set of pairwise non-crossing straight-line segments with 
  endpoints in~$V_\Pi$ that all lie below~$\Pi$.  Let~$v_k'$ be a
  point vertically below~$v_k$, let $\Pi'=v_1,\dots,v_{k-1},v_k'$, and
  finally let~$D'$ be a copy of~$D$ with each segment $v_i v_k \in D$
  replaced by~$v_i v_k'$.

  Then the segments in~$D'$ are pairwise non-crossing and all lie
  below~$\Pi'$.
\end{lemma}

\begin{proof}
  Let $v_{i_1} v_k, \dots, v_{i_m} v_k$ be the straight-line segments incident
  to~$v_k$ (both on the monotone path~$\Pi$ and in~$D$), sorted
  clockwise around~$v_k$ such that $v_{i_m} = v_{k-1}$.  Note that,
  since the straight-line segments in~$D$ are below~$\Pi$, the
  vertices $v_{i_1},\dots,v_{i_m}$ are also sorted according to increasing
  $x$-coordinates, and all of them have smaller $x$-coordinate than~$v_k$.  
  Hence, the situation is as depicted in
  \figurename~\ref{figure:moving:downwards:is:okay:for:chords}.
  
\begin{figure}[h]
  \centering
  \includegraphics{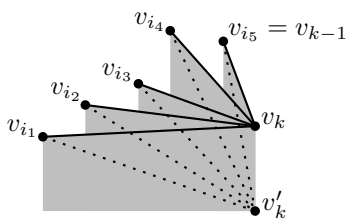}
  \caption{Illustration of the proof of
    Lemma~\ref{lemma:moving:downwards:is:okay:for:chords}.}
  \label{figure:moving:downwards:is:okay:for:chords}
\end{figure}

  For $1 \leq j \le m$, let~$B_j$ denote the set of points 
  that lie below the straight-line segment $v_{i_j} v_k$ and define 
  $B = \bigcup_{j=1}^m B_j$, see the shaded region
  in \figurename~\ref{figure:moving:downwards:is:okay:for:chords}.
  Note that the interior
  of~$B$ cannot contain any vertices of~$\Pi$ since this
  would contradict the fact that~$\Pi$ is $x$-monotone or the fact that
  all straight-line segments in~$D$ are below~$\Pi$.
  But this implies that none of the straight-line segments~$v_{i_j} v_k'$
  (drawn dotted in \figurename~\ref{figure:moving:downwards:is:okay:for:chords})
  is crossed by a straight-line segment in~$D'$ since this would
  yield a contradiction 
  to the fact that the straight-line segments in~$D$ are non-crossing or
  to the fact that the interior of~$B$ does not contain a vertex of~$\Pi$.
  No other crossings can occur in~$D'$ since the straight-line
  segments in~$D$ are non-crossing.  This finishes the proof.
\end{proof}

Recall that we aim at finding a large set $I \subseteq V_\Pi$ such
that no two vertices in~$I$ are connected by a chord of~$\Pi$.  The
set~$F$ of fixed vertices will be a subset of~$I$.  Note that~$I$ may
contain vertices connected by an \emph{edge} of~$\Pi$.  
In the following lemma, the set~$V^*$ contains all vertices of~$\Pi$
that we have to move in order to draw the chords of~$\Pi$
straight-line; clearly such a set must cover all chords of~$\Pi$.  
Thus the set~$V^*$ plays the role of the complement of~$I$.

\begin{lemma}
  \label{lemma:modification:procedure}
  Let $\Pi=v_1,\dots,v_k$ be an $x$-monotone polygonal path.
  Let~$C_\Pi$ be a set of chords of~$\Pi$ that can be drawn as
  non-crossing curved lines below~$\Pi$. 
  Let~$V^*$ be a vertex cover of~$C_\Pi$.  Then there is a way to
  modify~$\Pi$ by decreasing the $y$-coordinates of the vertices
  in~$V^*$ such that the resulting straight-line drawing~$\delta^*$ 
  of~$G_\Pi = (V_\Pi, E_\Pi \cup C_\Pi)$ is plane, the bounded faces
  of~$\delta^*$ are 
  star-shaped, and all edges in~$C_\Pi$ lie below the modified
  polygonal path, which is also $x$-monotone.  The coordinates of
  the vertices of the modified path have bit length~$O(nL)$, where~$L$
  is the maximum bit length of the vertex coordinates of~$\Pi$.
\end{lemma}

\begin{proof}
  We use induction on the number~$m=|C_\Pi|$ of chords.  If $m=0$, we need not
  modify~$\Pi$.  So, suppose that $m > 0$.  We first choose a chord $vw \in
  C_\Pi$ with $x(v) < x(w)$ such that there is no other edge $v'w' \in
  C_\Pi$ with the property that $x(v') \le x(v)$ and $x(w') \ge x(w)$.
  Clearly, such an edge always exists.  Then we apply the induction
  hypothesis to $C_\Pi \setminus \{vw\}$.  This yields a
  modification~$\Pi'$ of~$\Pi$ such that $\Pi'$ is $x$-monotone,
  all edges in the resulting straight-line 
  drawing~$\delta'$ of $G_\Pi - vw$ lie below~$\Pi'$, and all
  bounded faces in this drawing are star-shaped.  

  Now consider the chord~$vw$ and, without
  loss of generality, assume that $v \in V^*$. 
  Let~$Z$ denote the set of those vertices $z \in V_\Pi$
  with the property that no point vertically below~$z$ and distinct
  from~$z$ is contained in an edge of the drawing~$\delta'$.
  Note that, since~$\Pi'$ is $x$-monotone, there must exist a
  point~$p$ vertically  
  below~$v$ such that for no vertex~$z \in Z$, the straight-line
  segment~$pz$ crosses any edge in the drawing~$\delta'$ of $G_\Pi - vw$.

  Let $i \in \{1,\dots,k\}$ be such
  that $v = v_i$. We move vertex~$v$ to the point~$p$ to obtain
  a new drawing~$\delta''$ of $G_\Pi - vw$.  
  Then we apply Lemma~\ref{lemma:moving:downwards:is:okay:for:chords}
  to the rightmost vertex of the $x$-monotone subpath
  of~$\Pi'$ with vertices $v_1,\dots,v_i$ and, similarly,
  we apply Lemma~\ref{lemma:moving:downwards:is:okay:for:chords} to
  the leftmost vertex of the $x$-monotone subpath 
  of~$\Pi'$ with vertices $v_i,\dots,v_k$.  It follows that
  this does not produce any crossings among the edges in the
  drawing~$\delta''$.  
  Moreover, by our choice of the chord~$vw$, for each face in the
  drawing~$\delta'$ of $G_\Pi - vw$ that has vertex~$v$  
  in its facial cycle, $v$ must be the leftmost or the rightmost
  vertex in this facial cycle. 
  Hence, we can apply
  Lemma~\ref{lemma:moving:downwards:is:okay:for:kernel} to these
  faces.  This yields that they remain star-shaped in~$\delta''$.
  By the choice of~$p$ we ensure that the bounded face that results 
  from adding the straight-line segment~$pw$ to the
  drawing~$\delta''$ is also star-shaped.

  Concerning the size of the coordinates we argue as follows.  Without
  loss of generality we can assume that all vertices of~$\Pi$ have
  negative $y$-coordinates.  Now consider the addition of the $i$-th
  chord~$vw$.  Let~$y_{i-1}$ be the minimum $y$-coordinate of a vertex in
  the drawing before moving vertex~$v$ down.  Then it is not hard to
  check that, in order to add the chord~$vw$ without introducing any
  crossings, it suffices to move~$v$ down to a point with 
  $y$-coordinate $(2R_x)y_{i-1}$, where~$R_x$ is the ratio
  of the maximum over the minimum difference between the
  $x$-coordinates of any two distinct vertices in~$V_\Pi$.  
  Solving the recurrence for~$y_i$ yields $|y_i| \le |(2 R_x)^i y_0|$.
  Therefore,
  since there are only $O(n)$ chords, the $y$-coordinates in
  the resulting $x$-monotone path can be encoded using $O(nL)$
  bits.
\end{proof}

\begin{remark}
  Unfortunately, there are indeed instances where our algorithm
  actually needs $\Theta(n^2)$ bits for representing all $y$-coordinates
  of the modified path.  Let~$k>0$ be an odd integer, and let~$\Pi$ be
  a path with $n=2k+1$ vertices 
  $v_1,\dots,v_n$, where $v_i = (i,0)$ for $1 \le i \ne k+1 \le n$ and
  $v_{k+1} = (k+1,-1)$, see the thick light-gray path in
  \figurename~\ref{fig:coordinates}.  We set $C_\Pi = \{v_1v_n,
  v_2v_{n-1}, \dots, v_kv_{k+2}\}$ (drawn with dotted arcs in
  \figurename~\ref{fig:coordinates}) and $V^* = \{ v_2, v_4, \dots,
  v_{k-1}, v_{k+2}, \dots, v_{n-2}, v_n \}$ (marked with
  circles in \figurename~\ref{fig:coordinates}).

  Our algorithm straightens the chords in the order from innermost to
  outermost, that is, vertices are moved in the order $v_{k+2},
  v_{k-1}, v_{k+4}, \dots, v_2, v_n$.  To simplify presentation, let
  $w_1, w_2, \dots, w_k$ denote the vertices of~$V^*$ in this order,
  and let $w_0 = v_{k+1}$.  For $i=0,\dots,k$, denote the final position of
  $w_i$ by $(x_i,-y_i)$.  Then
  clearly $|x_i - x_{i-1}| = 2i-1$ for $i=1,\dots,k$.  The edges
  incident to~$w_{i-1}$ have slope $\pm y_{i-1}$ (with the exception
  of the irrelevant edge $w_0w_1$), thus $y_i > y_{i-1} + y_{i-1}
  \cdot |x_i - x_{i-1}| = y_{i-1} \cdot 2i$.  The recursion solves to
  $y_i > 2^i i!$.
\end{remark}

\begin{figure}
  \centering
  \includegraphics{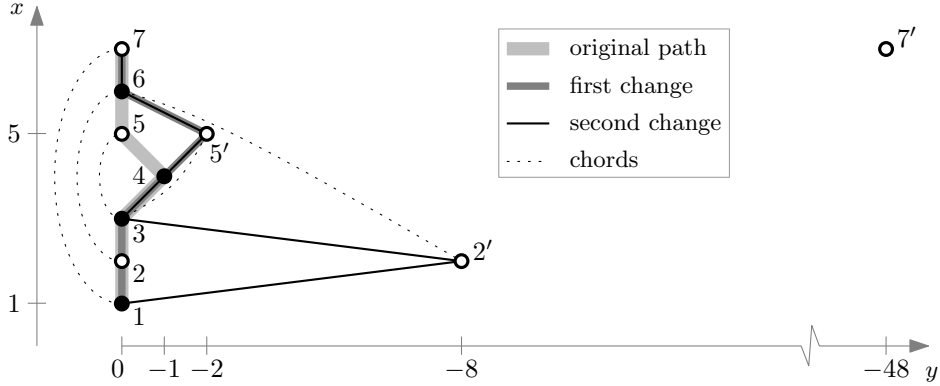}
  \caption{When applying the algorithm that proves
    Lemma~\ref{lemma:modification:procedure} to the thick light-gray
    $n$-vertex path with the dotted chords indicated on the left and
    with the vertex cover~$V^*$ indicated by circles, some of the
    $y$-coordinates of the resulting path need more than $n$ bits.
    Note that the $x$-axis is vertical.}
  \label{fig:coordinates}
\end{figure}

Now suppose that we have modified the $x$-monotone path~$\Pi_1$ according
to Lemma~\ref{lemma:modification:procedure}.  Then the resulting
$x$-monotone path~$\Pi_2$ admits a straight-line drawing of the chords
in~$C$ below~$\Pi_2$ such that the bounded faces are star-shaped
polygons, see the example in
\figurename~\ref{figure:steps:basic:idea}(d).  Recall that $u \in \Vtop$ is
the vertex of the outer triangle in~$\beta$ that does not lie
on~$\Pi$.  We now move vertex~$u$ to a position above~$\Pi_2$ such
that all edges $uw \in E$ with $w \in V_\Pi$ can be drawn without
crossing~$\Pi_2$ and such that the resulting faces are star-shaped
polygons.  Since~$\Pi_2$ is $x$-monotone, this can be done.  As an
intermediate result, we obtain a plane straight-line drawing of a
subgraph of~$G$ where all bounded faces are star-shaped.  It remains to
find suitable positions for the vertices in $(\Vtop \setminus \{u\})
\cup \Vbot$.  For every star-shaped face~$f$, there is a unique
subgraph~$G_f$ of~$G$ that must be drawn inside this face.  Note that
by our construction every edge of~$G_f$ that has both endpoints on the
boundary of~$f$ must actually be an edge of the boundary.  Therefore,
$G_f$ is a rooted triangulation where no facial cycle has a chord.
Now Lemma~\ref{lemma:wheellike:implies:triconnected} yields that~$G_f$
is triconnected.  Finally, we can use the result of Hong and Nagamochi
\cite{hn-cdgnc-08} (see Theorem~\ref{thm:hong-nagamochi}) to draw each
subgraph of type~$G_f$ and thus finish our construction of a plane
straight-line drawing of~$G$, see the example in
\figurename~\ref{figure:steps:basic:idea}(e).  We summarize.

\begin{theorem}
  \label{theorem:main}
  Let~$G$ be a triangulated planar graph that contains a simple path
  $\Pi=w_1,\dots,w_l$ and a face $uw_1w_l$.  If~$G$ has an
  embedding~$\beta$ such that $uw_1w_l$ is the outer face, $u$ lies on
  one side of~$\Pi$, and all chords of~$\Pi$ lie on the other side,
  then $\fix(G) \ge \sqrt{(l+1)/2}$.
\end{theorem}

\begin{proof}
  We continue to use the notation introduced earlier in this section.
  Recall that we aim at finding a large set~$I \subseteq V_\Pi$ such
  that no two vertices in~$I$ are connected by a chord of~$\Pi$.  The
  complement $V_\Pi \setminus I$ of this set~$I$ is the vertex
  cover~$V^*$ of~$C$ that we need for applying
  Lemma~\ref{lemma:modification:procedure}.

  Further, $F \subseteq I$ is the set of vertices that we 
  fixed in the first step, that is, in the construction of the
  $x$-monotone path~$\Pi_1$.  It follows from Proposition~1 in the
  paper by Pach and Tardos \cite{pt-up-02} that we can make sure that
  $\fix(G)=|F| \ge \sqrt{|I|}$.  Consider the graph~$G_C$ with vertex
  set~$V_\Pi$ and edge set~$C$.  An independent set
  in~$G_C$ has exactly the property that we want for~$I$.  Thus it
  suffices to show that the $l$-vertex graph~$G_C$ has 
  an independent set~$I$
  of size at least $(l+1)/2$.  We do this by giving a simple algorithm.

  Our algorithm is greedy: we always take a vertex~$v$ of smallest
  degree, put it in the independent set~$I$ under construction,
  remove~$v$ and the neighbors of~$v$ from~$V_\Pi$, and remove the
  edges incident to these vertices from~$C$.  We repeat this until~$G_C$ is
  empty.

  Note that~$G_C$ initially has at least one isolated (that is,
  degree-0) vertex, and that the bound is obvious if~$G_C$ is a
  forest---the algorithm first picks all isolated vertices and then
  repeatedly picks leaves.  Even if~$G_C$ contains cycles, the
  algorithm always picks vertices of degree at most~2.  This is due to
  the fact that all chords lie on one side of~$\Pi$, and thus~$G_C$ is
  and remains outerplanar, and any outerplanar graph has a vertex of
  degree at most~2.

  Let~$n_i$ be the number of vertices that have degree~$i$ when they
  are put in~$I$.  As observed above, $|I| = n_0 + n_1 + n_2$.  Whenever we put
  a vertex of degree~$i$ into~$I$, we remove $i+1$ vertices
  from~$V_\Pi$, thus $l = n_0 + 2 n_1 + 3 n_2$.  Let~$f$ be the number
  of bounded faces of~$G_C$.  Whenever the algorithm removes a
  degree-2 vertex, the number of bounded faces of~$G_C$ decreases by
  one, thus $f = n_2$.  We claim---and will prove below---that $f+1
  \le n_0$.  Now adding $1 \le n_0 - n_2$ to the above expression
  for~$l$ yields $l+1 \le 2 n_0 + 2 n_1 + 2 n_2 = 2 |I|$, or $|I| \ge
  (l+1)/2$, which proves the theorem.

  It remains to prove our claim, that is, $n_0 - 1 \ge f$.  In other
  words, we need to show that~$G_C$ contains at least one isolated
  vertex more than bounded faces.  Recall that~$G_C$ does \emph{not}
  include the edges of~$\Pi$.  For a chord $c=w_i w_j$ in~$C$, we
  define $\{w_{i},w_{i+1},\dots,w_{j}\} \subseteq V_\Pi$ to be the
  \emph{span} of~$c$.  Now consider a face~$F$ of~$G_C$ with vertices
  $w_{i_1},w_{i_2},\dots,w_{i_k}$ and $i_1 < i_2 < \cdots < i_k$.  The
  edges of~$F$ are $w_{i_1}w_{i_2}, w_{i_2}w_{i_3}, \dots,
  w_{i_k}w_{i_1}$.  Note that the span of $w_{i_k}w_{i_1}$ contains
  the span of every other edge of~$F$.  We define the span of~$F$ to
  be the span of the edge $w_{i_k}w_{i_1}$.

  We prove our claim by induction on~$f$.  As noted above, $G_C$
  contains at least one isolated vertex.  This establishes the base of
  the induction.  Now suppose that~$f>0$.  Consider the set~$M$ of all
  faces of~$G_C$ whose span is maximal with respect to set inclusion.
  If $|M|>1$, we apply the induction hypothesis to the subgraphs
  of~$G_C$ induced by the spans of the faces in~$M$.  Otherwise,
  let~$F^*$ be the only face in~$M$, and let $e_1,\dots,e_{k-1}$ be the
  edges of~$F^*$ whose span is properly contained in the span
  of~$F^*$.  We apply the induction hypothesis to the subgraphs
  of~$G_C$ induced by the spans of $e_1,\dots,e_{k-1}$.  Since $k \ge
  3$, there are at least two such subgraphs.  Each of them contains at
  least one isolated vertex more than bounded faces.  Taking~$F^*$
  into account, we conclude that~$G_C$ also contains at least one
  isolated vertex more than bounded faces.  This completes the proof
  of our claim.
\end{proof}

\subsection{Finding a Suitable Path}
\label{sub:finding:a:path}

We finally present two strategies for finding a suitable path~$\Pi$.
Neither depends on the geometry of the given
drawing~$\delta_0$ of~$G$.  Instead, they exploit the graph structure of~$G$.
The first strategy works well if~$G$ has a vertex of large degree and,
even though it is very simple, yields asymptotically tight bounds for
outerplanar graphs.

\begin{lemma}
  \label{lem:maxdeg}
  Let~$G$ be a triangulated planar graph with maximum degree~$\Delta$.
  Then $\fix(G) \ge \sqrt{(\Delta+1)/2}$.
\end{lemma}

\begin{proof}
  Let~$u$ be a vertex of degree~$\Delta$ and consider a plane
  embedding~$\beta$ of~$G$ where vertex~$u$ lies on the outer face.
  Since~$G$ is planar, such an embedding exists.  Let
  $\{w_1,\dots,w_{\Delta}\}$ be the set of neighbors of~$u$ sorted
  clockwise around~$u$ in~$\beta$.  This gives us the desired polygonal path
  $\Pi=w_1,\dots,w_{\Delta}$ that has no chords on the side that
  contains~$u$.  Thus Theorem~\ref{theorem:main} yields $\fix(G) \ge
  \sqrt{(\Delta+1)/2}$.
\end{proof}

Lemma~\ref{lem:maxdeg} yields a lower bound for outerplanar graphs
that is asymptotically tight as we will see in
Section~\ref{section:upper:bound:outerplanar}.

\begin{corollary}
  \label{cor:outerplanar}
  Let~$G$ be an outerplanar graph with~$n$ vertices.  Then $\fix(G)
  \ge \sqrt{n/2}$.
\end{corollary}

\begin{proof}
  We select an arbitrary vertex~$u$ of~$G$.  Since~$G$ is outerplanar,
  we can triangulate~$G$ in such a way that in the resulting
  triangulated planar graph~$G'$ vertex~$u$ is adjacent to every other
  vertex in~$G'$.  Thus the maximum degree of a vertex in~$G'$ is
  $n-1$, and the result follows by Lemma~\ref{lem:maxdeg}.
\end{proof}

Our second strategy works well if the diameter~$d$ of~$G$ is large.

\begin{lemma}
  \label{lemma:large:diameter}
  Let~$G$ be a triangulated planar graph of diameter~$d$.  Then
  $\fix(G) \ge \sqrt{d}$.
\end{lemma}

\begin{proof}
  We choose two vertices~$s$ and~$v$ such that a shortest $s$--$v$
  path has length~$d$.  We compute any plane embedding of~$G$ that
  has~$s$ on its outer face.  Let~$t$ and~$u$ be the neighbors of~$s$
  on the outer face.
  Recall that a \emph{Schnyder wood} (or
  \emph{realizer})~\cite{s-pgpd-89} of a triangulated plane graph is a
  (special) partition of the edge set into three spanning trees each
  rooted at a different vertex of the outer face.  Edges can be viewed
  as being directed to the corresponding roots.  The partition is
  special in that the cyclic pattern in which the spanning trees enter
  and leave a vertex is the same for all inner vertices.
  Schnyder~\cite{s-pgpd-89} showed that this cyclic pattern ensures
  that the three unique paths from a vertex to the three roots are
  vertex-disjoint and chordless.  Let~$\pi_s$, $\pi_t$, and~$\pi_u$ be
  the ``Schnyder paths'' from~$v$ to~$s$, $t$, and~$u$, respectively.
  Note that the length of~$\pi_s$ is at least~$d$, and the lengths
  of~$\pi_t$ and~$\pi_u$ are both at least $d - 1$.  Let~$\Pi$ be the
  path that goes from~$s$ along~$\pi_s$ to~$v$ and from~$v$
  along~$\pi_t$ to~$t$.  The length of~$\Pi$ is at least $2d-1$.  Note
  that, due to the existence of~$\pi_u$, the path~$\Pi$ has no chords on
  the side that contains~$u$.  Thus, Theorem~\ref{theorem:main} yields
  $\fix(G,\delta) \ge \sqrt{d}$.
\end{proof}

Next we determine the trade-off between the two strategies above.

\begin{theorem}
  \label{theorem:final:lower:bound}
  Let~$G$ be a planar graph with $n \ge 4$ vertices.  Then 
  $\fix(G) \ge \sqrt{\frac{(\log n)-1}{\log \log n}}$,
  where the base of logarithms is~$2$.
\end{theorem}

\begin{proof}
  Let~$G'$ be an arbitrary triangulation of~$G$.  Note that the
  maximum degree~$\Delta$ of~$G'$ is at least~$3$ since $n \geq 4$
  and~$G'$ is triangulated. 
  To relate~$\Delta$ to the diameter~$d$ of~$G'$, we use a very crude
  counting argument---Moore's bound: starting from an arbitrary vertex
  of~$G$, we bound the number of vertices we can reach by a path of a
  certain length.  Let~$j$ be the smallest integer such that $1 +
  (\Delta-1) + (\Delta-1)^2 + \dots + (\Delta-1)^j \ge n$.  Then $d
  \ge j$.  By the definition of~$j$ we have $n \le
  (\Delta-1)^{j+1}/(\Delta-2)$, which we can simplify to $n \le 2
  (\Delta-1)^j$ since $\Delta \ge 3$.  Hence we have $d \ge j \ge
  \frac{(\log n)-1}{\log (\Delta-1)}$.

  Now, if $\Delta \ge (\log n) + 2$, Lemma~\ref{lem:maxdeg} yields
  $\fix(G') \ge \sqrt{((\log n)+3)/2}$.  Otherwise $d \ge
  \frac{(\log n)-1}{\log \log n}$, and we can apply
  Lemma~\ref{lemma:large:diameter}.  Observing that $\fix(G) \ge
  \fix(G')$ yields the desired bound.
\end{proof}

\begin{remark}
  The proof of Theorem~\ref{theorem:final:lower:bound} (together with
  the auxiliary results stated earlier) yields an $O(n^2)$-time
  algorithm for untangling a given straight-line drawing of a planar
  graph~$G$ with~$n$ vertices by moving some of its vertices to new
  positions.  The first step, that is, computing the $x$-monotone
  path~$\Pi_1$, takes $O(n \log n)$ time~\cite{s-liads-61}.  Moving
  the vertices of~$\Pi_1$ such that the faces induced by the path and
  its chords become star-shaped takes $O(\gamma(n) n)$ time
  (Lemma~\ref{lemma:modification:procedure}), where $\gamma(n)=O(n)$
  is an upper bound on the time needed to perform an elementary
  operation involving numbers of bit length $O(n)$.  The remaining
  steps of our method can be implemented to run in $O(n)$ time.  This
  includes calling the algorithm of Hong and
  Nagamochi~\cite{hn-cdgnc-08} and computing the Schnyder
  wood~\cite{s-pgpd-89}, which we need in the proof of
  Lemma~\ref{lemma:large:diameter}.
\end{remark}

\section{Planar Graphs: Upper Bound} \label{sec:planar}

We now give an upper bound for general planar graphs that is better
than the upper bound $O((n \log n)^{2/3})$ of Pach and
Tardos~\cite{pt-up-02} for cycles.  Our construction uses the following
sequence, which we call~$\sigma_q$ and which we re-use in
Section~\ref{section:upper:bound:outerplanar}:
\[ \big((q-1)q, (q-2)q, \ldots, 2q, q, \underline{0}, 1+(q-1)q,
\ldots, 1+q, \underline{1}, \ldots, q^2-1, \ldots, (q-1)+q,
\underline{q-1}\big).\] %
Note that~$\sigma_q$ can be written as $(\sigma_q^0, \sigma_q^1, \dots
\sigma_q^{q-1})$, where $\sigma_q^i = ((q-1)q+i, (q-2)q+i, \ldots,
2q+i, q+i, i)$ is the subsequence of~$\sigma_q$ that consists of
all~$q$ numbers in~$\sigma_q$ that are congruent to~$i$ modulo~$q$.
To stress this, the last element in each of these subsequences is
underlined in~$\sigma_q$.
Thus~$\sigma_q$ consists of~$q^2$ distinct numbers.  It is not hard to
see the following.

\begin{observation}
  \label{observation:structure:sigma:one}
  The longest increasing or decreasing subsequence of~$\sigma_q$ has
  length~$q$.
\end{observation}

We call two subsequences $\Sigma=s_1,s_2,\dots,s_{l}$
and $\Sigma'=s'_1,s'_2,\dots,s'_{l'}$ of~$\sigma_q$ 
\emph{separated} if 
\begin{enumerate}[(i)]
\item $s_{l}$ comes before~$s'_1$ or
      $s'_{l'}$ comes before~$s_1$ in~$\sigma_q$, and
\item $\max(\Sigma)  < \min(\Sigma')$ or 
      $\max(\Sigma') < \min(\Sigma$).
\end{enumerate}

\begin{observation}
  \label{observation:structure:sigma:two}
  Let~$\Sigma$ and~$\Sigma'$ be two separated decreasing or two
  separated increasing subsequences of~$\sigma_q$.  Then
  $|\Sigma \cup \Sigma'| \le q+1$.
\end{observation}
\begin{proof}
  First consider the case that~$\Sigma$ and~$\Sigma'$ are both
  decreasing.  Since they are separated we can assume without loss of
  generality that $\max(\Sigma) < \min(\Sigma')$.  We define $V_i =
  \{iq + j : 0 \leq j \leq q-1\}$ for $i=0,\dots,q-1$.  Then, since
  $\Sigma$ and~$\Sigma'$ are both decreasing, they can each have at
  most one element in common with every~$V_i$.  Now suppose that they have
  both one element in common with some~$V_{i_0}$.  Then, since
  $\max(\Sigma) < \min(\Sigma')$, $\Sigma$ cannot have an element in
  common with any~$V_i$, $i > i_0$, and $\Sigma'$ cannot have an
  element in common with any~$V_i$, $i < i_0$.  Therefore,~$|\Sigma
  \cup \Sigma'| \leq q + 1$.

  If~$\Sigma$ and~$\Sigma'$ are both increasing, then, similarly as
  above, every subsequence~$\sigma_q^i$ with $0 \leq i \leq q-1$
  can have at most one element in common with each~$\Sigma$
  and~$\Sigma'$. Moreover, at most one subsequence~$\sigma_q^i$ 
  can have an element in common with both~$\Sigma$ and~$\Sigma'$. 
  This implies that $|\Sigma \cup \Sigma'| \leq q + 1$, as required.
\end{proof}

\begin{theorem} \label{thm:planar-upper} %
  For any integer $n_0>0$, there exists a planar graph \G\ with $n
  \ge n_0$ vertices and $\fix(\G) \le \sqrt{n-2}+1$.
\end{theorem}

\begin{proof}
  Let $q = \lceil \sqrt{n_0} \rceil$.  We define the graph~\G\ as a
  path of~$q^2$ vertices~$1, 2, \dots, q^2$ all connected to the two
  endpoints of an edge~$\{a,b\}$ with $a, b \not\in \{1, 2, \dots, q^2\}$,
  see \figurename~\subref*{sfg:stacked}. Hence~$\G$ has $n=q^2+2$
  vertices.  Let~$\dG$ be the drawing of~$\G$ 
  where vertices~$1, 2, \dots, q^2$ are placed on a vertical
  line~$\ell$ in the order given by~$\sigma_q$.  We place vertices~$a$
  and~$b$ below the others on~$\ell$, see 
  \figurename~\subref*{sfg:convex}.
  \begin{figure}[tb]
    \subfloat[plane drawing of case~1\label{sfg:stacked}]%
    {\parbox[b]{.33\textwidth}{\centering
        \includegraphics[scale=.8]{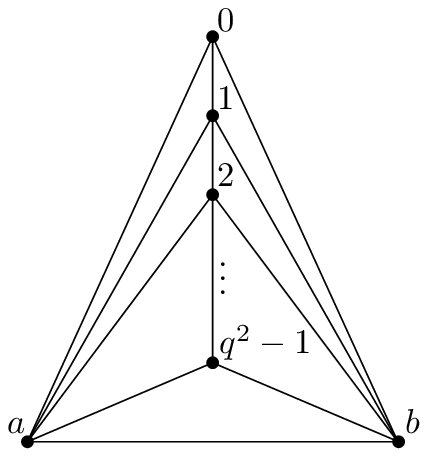}}}
    \hfill %
    \subfloat[drawing~$\dG$ (without edges)\label{sfg:convex}]%
    {\parbox[b]{.31\textwidth}{\centering
        \includegraphics[scale=.8]{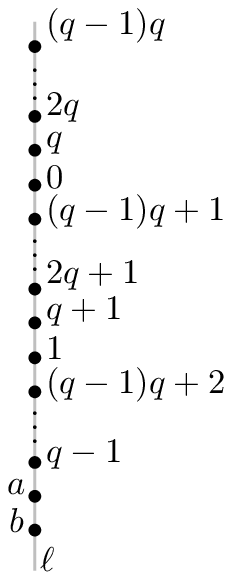}}}
    \hfill %
    \subfloat[plane drawing of case~2\label{sfg:case-two}]%
    {\parbox[b]{.33\textwidth}{\centering
        \includegraphics[scale=.8]{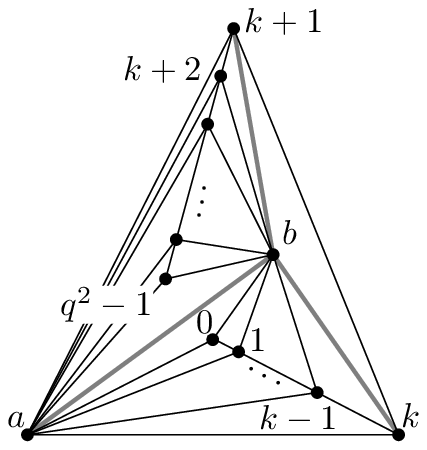}}}
    \caption{Drawings of the graph \G\ that we use in the proof of
      Theorem~\ref{thm:planar-upper}.} 
    \label{fig:planar-upper}
  \end{figure}
  Let~$\dG'$ be an arbitrary plane drawing of~$\G$ obtained
  by untangling~$\dG$.
  Since all faces of~$\G$
  are 3-cycles, the outer face in~$\dG'$ is a triangle.  All faces
  of~$\G$ contain~$a$ or~$b$.  This has two consequences.  First, $a$
  and~$b$ must move to new positions in~$\dG'$, otherwise all
  other vertices would have to move.  Second, at least one of them,
  say~$a$, appears on the outer face.

  \smallskip
  \noindent
  \emph{Case 1:} Vertex~$b$ also lies on the outer face. 

  \noindent
  Then there are just two possibilities for the embedding of~$\G$: as
  in \figurename~\subref*{sfg:stacked} or with the indices of all vertices
  reversed, that is, vertex~$i$ becomes $q^2-i-1$.  
%
  Now let $0 \leq i < j < k \leq q^2-1$ be three fixed vertices.  By
  symmetry we can assume that~$j$ lies in $\Delta(a,b,i)$.  Then~$k$
  also lies in $\Delta(a,b,i)$ since the path connecting~$j$ to~$k$
  does not intersect the sides of this triangle.  Note that~$k$ cannot
  lie between~$i$ and~$j$ on~$\ell$ as otherwise one of the edges
  $\{a,k\}$ and $\{b,k\}$ would intersect the polygonal path
  connecting~$i$ to~$j$.  Thus, each triplet of fixed vertices forms a
  monotone sequence along~$\ell$.  This in turn yields that \emph{all}
  fixed vertices in $\{0, \ldots, q^2-1\}$ form a monotone sequence
  along~$\ell$.  Due to the construction of~$\sigma_q$, such a sequence
  has length at most $q=\sqrt{n-2}$.

  \smallskip
  \noindent
  \emph{Case 2:} Vertex~$b$ does not lie on the outer face.

  \noindent
  Then the outer face is of the form $\Delta(a,k,k+1)$ with $0 \leq k
  \leq q^2-2$. The three edges $\{b,a\}$, $\{b,k\}$, and $\{b,k+1\}$
  incident to~$b$ split $\Delta(a, k, k+1)$ into the three triangles
  $\Delta(a,k,b)$, $\Delta(a,b,k+1)$, and $\Delta(b,k,k+1)$, see
  \figurename~\subref*{sfg:case-two}.  Every vertex of~$\dG'$ lies
  in one of them.  Since~$\dG'$ is plane, vertex~$k-1$ must belong
  to $\Delta(a,k,b)$, and, by induction, so do all vertices $i \leq k$;
  similarly, all vertices $i \geq k+1$ lie in $\Delta(a,b,k+1)$. We
  can thus apply the argument of case~1 to each of the two subgraphs
  contained in $\Delta(a,b,k)$ and $\Delta(a,b,k+1)$.  This yields two
  separated monotone sequences of length at most~$q$ each.  Note,
  however, that both are increasing or both are decreasing since one
  type forces~$a$ to the left and~$b$ to the right of~$\ell$ and the
  other does the opposite. Due to
  Observation~\ref{observation:structure:sigma:two}, the length of two
  separated monotone subsequences of~$\sigma_q$ sums up to at
  most $q+1=\sqrt{n-2}+1$.

  \smallskip\noindent
  To summarize, case~2 yields a larger number of potentially fixed
  vertices, and thus $\fix(\G, \dG) \le q+1=\sqrt{n-2}+1$. 

  Note that actually
  $\fix(\G, \dG) = q+1$ as we can fix, for example, the vertices 
  $0,q,2q,\dots,(q-1)q$, and $(q-1)q+2$.
\end{proof}

\section{An Upper Bound for Outerplanar Graphs}
\label{section:upper:bound:outerplanar}

In this section we show that the lower bound $\fix(\HH) \ge \sqrt{n/2}$
that holds for any outerplanar graph~$\HH$ with~$n$ vertices (see
Corollary~\ref{cor:outerplanar}) is asymptotically tight in the worst case. 

\begin{figure}[tb]
  \hfill
  \subfloat[\label{sfg:outerplanar-1}]{\fbox{\includegraphics{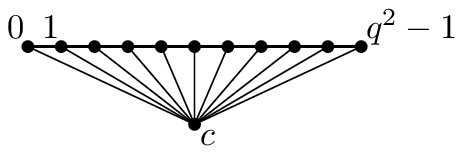}}}
  \hfill%
  \subfloat[\label{sfg:outerplanar-2}]{\fbox{\includegraphics{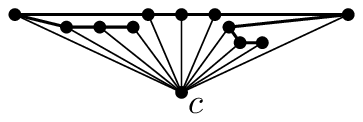}}}
  \hfill\hspace*{0pt}

  \caption{The outerplanar graph~$\HH$ that we use in the proof
    of Theorem~\ref{thm:outerplanar}.}
  \label{fig:outerplanar}
\end{figure}

\begin{theorem}
  \label{thm:outerplanar}
  For any integer~$n_0$, there exists an outerplanar graph~$H$ with 
  $n \ge n_0$ vertices and $\fix(H) \le 2\sqrt{n-1} + 1$.
\end{theorem}

\begin{proof}
  Let $q = \lceil \sqrt{n_0} \rceil$.  We define the outerplanar graph~\HH\
  as a path of~$q^2$ vertices $0,1,\dots,q^2-1$ and an extra vertex $c=q^2$
  that is connected to all other vertices, see
  \figurename~\subref*{sfg:outerplanar-1}.  Hence~\HH\ has $n=q^2+1$ vertices.
  Let~$\dH$ be the drawing of~$\HH$ where all vertices are placed
  on a horizontal line~$\ell$ as follows.
  Vertices~$0,\dots,q^2-1$ are arranged in the
  order~$\sigma_q$ introduced in Section~\ref{sec:planar}, and
  vertex~$c$ can go to an arbitrary (free) spot on~$\ell$.

  In the following we show that 
  $\fix(\HH,\dH) \le 2q+1 = 2\sqrt{n-1} +1$. 
  To this end,  let~$\dH'$ be an arbitrary plane drawing of~$\HH$ 
  obtained by untangling~$\dH$, and let~$F$ be the set of fixed vertices.
  Note that~$\HH$ has many plane
  embeddings---for example, \figurename~\subref*{sfg:outerplanar-2}---but only
  two outerplane embeddings: \figurename~\subref*{sfg:outerplanar-1} and its
  mirror image. 
  Our proof exploits the fact that the simple
  structure of~$\HH$ forces the left-to-right sequence of the fixed
  vertices to also have a very simple structure.

  Consider the drawing~$\dH'$.  If vertex~$c$ lies on~$\ell$ in~$\dH'$, then,
  since~$c$ is connected by an edge to every other vertex of~$\HH$
  and all these vertices lie on~$\ell$ in the drawing~$\dH$,
  at most two of these other vertices can be fixed. Hence, the interesting case
  is that~$c$ does not lie on~$\ell$ in~$\dH'$ and, therefore,
  $c \not \in F$. Hence, $F \subseteq \{0,1,\dots,q^2-1\}$.
  We only consider the interesting case that $|F| \geq 2$.
  Let~$m$ and~$M$ be the minimum and maximum in~$F$, respectively. 
  Without loss of generality we assume that~$c$ lies below~$\ell$
  and that~$m$ lies to the left of~$M$ (otherwise we reflect~$\dH'$ on the
  $x$/$y$-axis).  Let~$a$ and~$b$ be the left- and rightmost
  vertices in~$F$, see \figurename~\subref*{sfg:good}.

  \begin{figure}[htb]
    \centering \subfloat[good
    sequence\label{sfg:good}]{\includegraphics{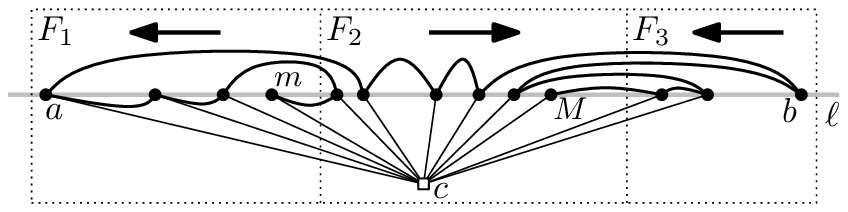}} \hfill
    \subfloat[bad
    sequence\label{sfg:bad}]{\includegraphics{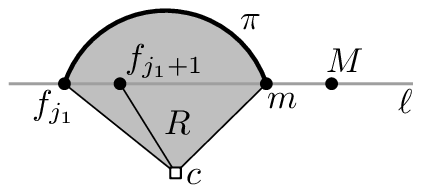}}
    \caption{Analyzing the sequence of fixed vertices along the line
      $\ell$.}
    \label{fig:sequences}
  \end{figure}

  Let $F_0 = f_1,f_2,\dots,f_{|F|}$ be the vertices in~$F$ ordered as we
  meet them along~$\ell$ from left to right.
  Let $F_1 = f_1,f_2,\dots,f_{j_1}$ be the longest subsequence of~$F_0$ 
  starting at~$f_1$ such that $f_{i-1} > f_i$ for $2 \leq i \leq j_1$. 
  Note that by definition $f_1 = a$.  We claim that $f_{j_1}=m$.  
  Assume to the contrary that $f_{j_1} \ne m$. 
  Then $f_{j_1} > m$ and, clearly, $F_1$ does not contain~$m$.  
  Thus~$m$ lies to the right of~$f_{j_1+1}$.

  Consider the path $\pi = f_{j_1}, f_{j_1}-1,\dots,m$ in~$\HH$. 
  Since $f_{j_1+1} > f_{j_1} > m$, $f_{j_1+1}$
  is not a vertex of~$\pi$. Let~$R$ be the polygon bounded
  by~$\pi$ and by the edges~$cf_{j_1}$ and~$cm$. Since~$\dH'$ is plane,
  $R$ is simple. Note that~$f_{j_1+1}$ lies in the interior of~$R$
  and~$M$ lies in the exterior of~$R$, as indicated in
  \figurename~\subref*{sfg:bad},  
  where the interior of~$R$ is shaded. To see this, 
  first note that~$M$, which lies to the 
  right of~$m$, cannot lie in the interior of~$R$ since otherwise the
  path in~$\HH$ with vertices $M,M-1,\dots,f_{j_1}$ would intersect
  $\pi$ or one of the edges that connect a vertex on~$\pi$
  with~$c$. But then, since~$\pi$ can intersect neither edge~$cM$ nor edge 
  $cf_{j_1+1}$, vertex~$f_{j_1+1}$ must lie in the interior of~$R$, as required.
  This yields a contradiction since the path 
  $M,M-1,\dots,f_{j_1+1}$ in~$\HH$ must now intersect the
  boundary of~$R$. Thus our assumption $f_{j_1} \ne m$ is
  wrong, and we have indeed $f_{j_1}=m$.

  Now let $F_2 = f_{j_1+1},f_{j_1+2},\dots,f_{j_2}$ be the longest 
  subsequence of~$F_0$ starting at~$f_{j_1+1}$ such that
  $f_{i-1} < f_i$ for each~$i$ with $j_1+1 \leq i \leq j_2$. 
  With similar arguments as above we can show that~$f_{j_2}=M$.  
  Moreover, let $F_3=f_{j_2+1},f_{j_2+2},\dots,f_{j_3}$ be the subsequence 
  of~$F_0$ starting at~$f_{j_2+1}$ such that
  $f_{i-1} > f_i$ for each~$i$ with $j_2+1 \leq i \leq j_3$. 
  Again, with similar arguments as above we can show that either~$F_3$
  is empty or 
  $f_{j_3} = f_{|F|}$. In addition, we can show in an analogous way that
  $f_1 < f_{|F|}$ holds.

  Thus, the set~$F$ is partitioned into~$F_1$, $F_2$, and~$F_3$. The
  sequence~$F_2$ is increasing, and both~$F_1$ and~$F_3$ are
  decreasing (or empty).  Thus, by
  Observation~\ref{observation:structure:sigma:one}, $|F_2| \le q$ and, 
  by Observation~\ref{observation:structure:sigma:two},  $|F_1|+|F_3| \le q+1$,
  since $f_1 < f_{|F|}$ implies that~$F_1$ and~$F_3$ are separated.
  Hence, $|F| \le 2q+1$, as required. 

  Note that this upper bound is
  almost tight: $\fix(\HH,\dH) \geq 2q - 2$ as indicated in 
  \figurename~\ref{figure:fixed:vertices:outerplanar}.
\end{proof}

\begin{figure}[h]
  \centering
  \includegraphics{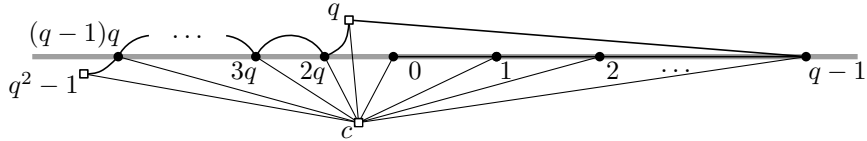}
  \caption{This plane drawing of the graph~$\HH$ 
    defined in the proof of Theorem~\ref{thm:outerplanar}
    shows that $\fix(\HH,\dH) \geq 2q - 2$ since it fixes the vertices
    $0,1,2,\dots,q-1,2q,3q,\dots,(q-1)q$ of~$\dH$.  The curved
    arcs indicate chains of vertices that have been moved.}
  \label{figure:fixed:vertices:outerplanar}
\end{figure}

\section{Conclusions}
\label{section:conclusions}

In this paper, we have presented several new results on the problem of
untangling a given drawing of a graph, a problem originally introduced
by Watanabe~\cite{w-op-98} for the special case of cycles.

On the computational side, we have proved that \textsc{MinShiftedVertices}
is NP-hard and also hard to approximate; we also showed that our proof
technique extends to another graph drawing problem, namely
\textsc{1BendPointSetEmbeddability} with given vertex--point
correspondence.  Related questions that remain
open are the inapproximability of \textsc{MaxFixedVertices} and the
hardness of \textsc{MaxFixedVertices} and \textsc{MinShiftedVertices}
for special classes of graphs such as cycles.  We have shown that all
these problems lie in \pspace, but do they also lie in~$\cal NP$?
Also, we are not aware of any result in the direction of parameterized
complexity.

On the combinatorial side, Table~\ref{table:summary} summarizes the
best currently known worst-case bounds for untangling several
important classes of planar graphs.
It reveals that the gap for general planar graphs is probably the most
interesting remaining open problem in the field.

\addcontentsline{toc}{section}{\numberline{}Acknowledgments}
\section*{Acknowledgments}

We thank Vida Dujmovi\'{c} for her comments on an earlier version of
this paper.  She came up with an idea that directly improved our lower
bounds in Theorem~\ref{theorem:final:lower:bound} and
Corollary~\ref{cor:outerplanar} by a factor of~$\sqrt{3}$.  Thanks to
Tom{\'a}{\v s} Gaven{\v c}iak for interesting discussions about the
results of Kang et al.\ \cite{kprsc-odpg-08} and to Ignaz Rutter
for the idea behind Theorem~\ref{thm:1bend-pspace}.  We thank all
anonymous and the non-anonymous referee for their detailed comments
that have helped us to improve the presentation of the paper.

\addcontentsline{toc}{section}{\numberline{}\refname}
\bibliographystyle{alpha}
\bibliography{abbrv,geometric_representation,packing+covering,np,comb_geometry,graph_drawing,graph_theory,vertex_move,morphing,posets}

\end{document}